\documentclass{article}
\usepackage[a4paper, total={6in, 9in}]{geometry}
\usepackage{xr}
\usepackage{hyperref}
\pdfoutput=1
\usepackage[utf8]{inputenc}
\usepackage{titling}
\usepackage{amsmath}
\usepackage[capitalise]{cleveref}
\usepackage{amssymb}
\usepackage{amsthm}
\usepackage{cite}
\usepackage{xcolor}
\usepackage{graphicx}
\usepackage{subcaption}
\usepackage{float}
\usepackage{paralist,enumerate}
\usepackage{appendix}
\usepackage{comment}
\usepackage{titlesec}
\usepackage{mathabx}
\usepackage{bm}
\usepackage[ruled,vlined]{algorithm2e}
\SetKw{KwBy}{by}
\SetAlFnt{\small}
\titlelabel{\thetitle.\quad}
\definecolor{azure}{rgb}{0.0, 0.5, 1.0}

\newtheorem{proposition}{Proposition}
\theoremstyle{remark}
\newtheorem*{remark}{Remark}

\title{Rank-based Bayesian clustering via covariate-informed Mallows mixtures}
\author{Emilie Eliseussen$^1$ \and Arnoldo Frigessi$^1$ \and Valeria Vitelli$^1$}
\date{%
    $^1$Oslo Centre for Biostatistics and Epidemiology, Department of Biostatistics, University of Oslo, Oslo, Norway\\[2ex]%
    E-mail: valeria.vitelli@medisin.uio.no\\[2ex]%
    February 16, 2024
}

\begin{document}

\maketitle

\begin{abstract}
    Data in the form of rankings, ratings, pair comparisons or clicks are frequently collected in diverse fields, from marketing to politics, to understand assessors' individual preferences. Combining such preference data with features associated with the assessors can lead to a better understanding of the assessors' behaviors and choices. The Mallows model is a popular model for rankings, as it flexibly adapts to different types of preference data, and the previously proposed Bayesian Mallows Model (BMM) offers a computationally efficient framework for Bayesian inference, also allowing capturing the users' heterogeneity via a finite mixture. We develop a Bayesian Mallows-based finite mixture model that performs clustering while also accounting for assessor-related features, called the Bayesian Mallows model with covariates (BMMx). BMMx is based on a similarity function that a priori favours the aggregation of assessors into a cluster when their covariates are similar, using the Product Partition models (PPMx) proposal. We present two approaches to measure the covariate similarity: one based on a novel deterministic function measuring the covariates' goodness-of-fit to the cluster, and one based on an augmented model as in PPMx. We investigate the performance of BMMx in both simulation experiments and real-data examples, showing the method's potential for advancing the understanding of assessor preferences and behaviors in different applications. 
    
\end{abstract}

\section{Introduction}\label{sec:intro}

Preference data are frequently encountered and utilized to assess the preferences of individuals, here referred to as "assessors". Assessors generate preference data in the form of (incomplete) rankings, ratings, pair comparisons, or by clicking on proposed items, often resulting in large sparse datasets. The primary objective is to uncover the missing parts of each assessor's ranking and identify the shared consensus ranking among assessors. However, this task is challenging due to the vast parameter space encompassing all possible permutations of the items ranked. The Mallows model \cite{mallows1957} has emerged as a promising choice to address this task, and recent advancements, such as the Bayesian Mallows model \cite{vitelli2018}, have introduced methods that enable to address the associated challenges in the utilization of arbitrary right-invariant distances, enhancing the model's flexibility. In this framework, the unknown parameters are the latent orderings of all items for each individual assessor, and the latent consensus ranking shared by all assessors in a homogeneous group. However, assuming a single consensus ranking for all assessors in a large pool is often unrealistic, as assessors differ group-wise in their preferences. Within each group of assessors with alike preferences one can assume that a latent ranking of items is shared. Consequently, it is necessary to partition assessors into homogeneous groups and perform inference on the model parameters within each group. For this purpose, previous research has proposed finite-mixture clustering methods \cite{vitelli2018, sorensen2019}. 

Recently, there has been a growing recognition of the use of rank-based models in data integration, as they offer a valuable approach not only for understanding individual behaviors but also for integrating diverse datasets \cite{deng2014, asfari2014, badgley2015, shang2020}.
In this paper, we wish to integrate in the analysis additional features that are often available about the assessors and that can help inference on the clustering of assessors. Covariate information, such as age, geographic location, gender, device used for and timing of the expressed preferences, and other behavioral traits, is known about each assessor, albeit potentially incomplete. Combining preference data with covariate information about the assessors can lead to a better understanding of the assessor’s behaviors and latent preferences. 

In this paper, we address the gap in existing methodologies to embed covariate information into rank-based clustering, by developing a novel Mallows-based finite mixture model, which we call the Bayesian Mallows model with covariates (BMMx). We propose a model-based method that a priori favours the aggregation of assessors into clusters, based on their covariate similarity. Rank data then enter the likelihood of the model and further influence the partition of the assessors. We adapt the Product Partition model with covariates (PPMx) \cite{muller2011} approach to our specific context, by eliciting covariate-dependent priors to explore various variants and alternatives. Specifically, we develop two ways to embed covariate similarity into the prior: one based on a novel deterministic function measuring the covariates' goodness-of-fit to the cluster, and one based on an augmented model. The posterior distribution, informed by the preference data expressed by each assessor, strikes a balance between covariate similarity and behavioral similarity. 

In the literature, there are some approaches that allow the integration of covariates with ranking data. One is the Plackett-Luce model \cite{luce1959, plackett1975}, as demonstrated in \cite{turner2020}, where a method for model-based clustering is presented. However, the Plackett-Luce model is less flexible with respect to different types of data-generating schemes. Li et al. \cite{li2022} develop an infinite mixture of Thurstone models to deal with available covariates for the items, however this mixture approach does not handle assessors' covariates, and moreover it does not estimate the latent individual ranking of each assessor. The Partition Mallows Model (PAMA) \cite{zhu2021} estimates the quality of each assessor, potentially assisted by covariate information if available. This model is based on the Bayesian algorithm for rank data (BARD) \cite{deng2014}, which assigns aggregated ranks based on the posterior probability that a certain item is relevant. However, PAMA works only for the Kendall distance in the Mallows model, in which case the partition function has  a closed form. In contrast, our approach can accommodate any right invariant distance measure. Additionally, PAMA does not perform clustering of the assessors. Gormley and Murphy \cite{gormley_murphy2008} propose a mixture of experts model for rank data, applying it to cluster Irish voters into similar voting preferences, using their votes as well as the associated covariates. A finite mixture model is implemented, where the influence of covariates on the mixing proportions is modeled via a generalized linear model estimated via an EM algorithm. Alternative not-rank-based methods to include covariates in clustering frameworks have also been developed, such as iCluster \cite{mo2017}, an integrative clustering method for genomic data based on a latent Gaussian mixture, where covariates can be included. Covariate-dependent multivariate functional clustering \cite{yang2022}  estimates a lower-dimensional representation of both the (functional) observations and covariates via a sparse latent factor model. Here the clustering model is assumed to depend only on a latent random effect, modeled with a Dirichlet Process (DP) mixture of Gaussian distributions. Note that both these two latter approaches assume equal influence of data and covariates on the clustering, as covariates are included and weighted equally to the data in the model (via an additional information layer and an additional factor in the model, respectively). This is quite different from our approach, where covariates instead inform the clustering a priori. Additional work on product partition models focuses on studying different ways to measure the covariates' similarity, with several approaches proposed in \cite{page_quintana2018}. However, the primary motivation behind this work is to address issues associated with using a Dirichlet Process prior when the number of covariates increases, resulting in very large or singleton clusters. These effects have not been observed in our finite mixture implementation.

Our proposed approach offers a comprehensive inferential procedure that uses both rankings and assessor-related covariate information, and that integrates clustering and preference learning into a unified framework. Operating within a probabilistic Bayesian setting, our method facilitates the computation of complex probabilities of interest, allowing for a nuanced and transparent analysis of the data. Our main contribution lies in developing the first Mallows-based method that incorporates assessor-related covariates. This novel approach holds significant potential for advancing the understanding of assessor preferences and behaviors. 

The structure of the paper is as follows: Section \ref{sec:model} describes the Bayesian Mallows model with covariates (BMMx) and presents the different measures of covariate similarity. The inferential procedure associated to BMMx is outlined in Section \ref{sec:mcmc}. To evaluate the performance of our method, we conduct several simulation studies, which are described in Section \ref{sec:simulations}. Furthermore, we demonstrate the versatility of the BMMx model in two real-world case studies presented in Section \ref{sec:case_study}. Finally, we summarize our findings and discuss potential future extensions of the model in Section \ref{sec:discussion}.

\section{Bayesian Mallows model with covariates (BMMx)}\label{sec:model}
A useful model for ranking data is the Bayesian Mallows model, which has been introduced in \cite{vitelli2018}. To incorporate covariate information into the Bayesian Mallows model, we let the covariates a priori inform the clustering by introducing a similarity function that favors the aggregation of assessors into a cluster when the covariates are similar. This approach is inspired by the Product Partition models (PPMx), however we develop it further by introducing two ways to measure the covariate similarity: one based on an augmented model similar to PPMx, and one based on a novel deterministic function measuring the covariates' goodness-of-fit to the cluster. 
We briefly recall the Bayesian Mallows model in the case of complete data in Section \ref{sec:model_bmm}, before introducing the covariate-dependent Bayesian Mallows clustering scheme with its variations in Section \ref{sec:model_bmmx}, as well as the choice of the similarity function in Section \ref{sec:model_bmmx_simfun}.

\subsection{Bayesian Mallows model (BMM)}\label{sec:model_bmm}
Consider a set of $n$ items denoted as $\mathcal{A}=\{A_1,A_2,\ldots,A_n\}$. We assume that all $n$ items are ranked according to a specified feature by $N$ assessors, thus providing complete rankings $\mathbf{R}_j= \{\mathbf{R}_{1j},\mathbf{R}_{2j},\ldots, \mathbf{R}_{nj}\}$, $j=1,\ldots,N$. The Mallows model \cite{mallows1957} is a probabilistic model for rankings defined on the space $\mathcal{P}_n$ of permutations of dimension $n$, taking the form
\begin{equation}\label{eq:Mallows}
    P(\bm{R} | \alpha, \bm{\rho}) = \frac{1}{Z_n(\alpha, \bm{\rho})}\exp\left\{-\frac{\alpha}{n} d(\bm{R}, \bm{\rho})\right\}1_{\mathcal{P}_n}(\bm{R}),
\end{equation}
where $\alpha > 0$ is a positive scale parameter, $\bm{\rho} \in \mathcal{P}_n$ is the latent consensus ranking, $Z_n(\alpha, \bm{\rho})$ is the partition function, $1_S(\cdot)$ is the indicator of the set $S$, and $d(\bm{R}, \bm{\rho})$ is a distance between $\bm{r}\in \mathcal{P}_n$ and $\bm{\rho}$.  

There are several choices for the distance function. For right-invariant distances, the partition function does not depend on  $\bm{\rho}$, since a right-invariant metric $d(\cdot,\cdot)$ is such that, for any $r_1, r_2 \in \mathcal{P}_n$, it holds: $d(r_1,r_2)=d(r_1 r_2^{-1}, \bm{1}_n)$, $\bm{1}_n= \{1,2,\ldots,n\}$ \cite{diaconis1988}. Thus, for right-invariant metrics, $Z_n(\alpha, \bm{\rho})=Z_n(\alpha)$, which gives important computational advantages. The partition function can be computed analytically for the Kendall distance, but this is not the case for most other distances, such as the footrule distance and the Spearman distance, and $Z_n(\alpha)$ must be approximated for inference. In this paper, we will use the footrule distance, which is defined as $d(\bm{R}, \bm{\rho})=\sum_{i=1}^n |R_i - \rho_i|$, the equivalent of an $\ell^1$ measure between rankings, but other options can easily be used. Note that, when the footrule distance is assumed, the scale parameter $\alpha$ can be interpreted analogously to the ``precision'' in a Gaussian model, as it quantifies the spread of the rankings around the ``mean'' consensus ranking $\bm{\rho}$. Several ways for approximating the model partition function when choosing the footrule or other distances have been described in \cite{vitelli2018}. We will use the same approaches here. 

The likelihood associated to the observed rankings $\mathbf{R}_1, \ldots, \mathbf{R}_N$ under the Mallows model is 
\begin{equation}
    P(\bm{R}_1, \ldots, \bm{R}_N | \alpha, \bm{\rho}) = \frac{1}{Z_n(\alpha)^N} \exp \left\{ -\frac{\alpha}{n} \sum_{j=1}^N d(\bm{R}_j, \bm{\rho})\right\} \prod_{j=1}^N 1_{\mathcal{P}_n}(\bm{R}_j).
\end{equation}

We specify priors for $\alpha$ and $\bm{\rho}$ as in \cite{vitelli2018} with a truncated exponential prior for $\alpha$ and a uniform prior for $\bm{\rho}$ in the space $\mathcal{P}_n$ of $n$-dimensional permutations:
\begin{equation}\label{eq:bmm_priors}
    \pi(\alpha) = \lambda e^{-\lambda\alpha} \frac{1_{[0, \alpha_{\max}]}(\alpha)}{(1-e^{-\lambda \alpha_{\max}})}, \hspace{1cm} \pi(\bm{\rho}) = \frac{1}{n!}1_{\mathcal{P}_n}(\bm{\rho})
\end{equation} 
where $\alpha_{\max} < \infty$ is a cut-off point and is large compared to the value of $\alpha$ supported by the data. The posterior distribution for $\bm{\rho}$ and $\alpha$ then becomes
\begin{equation}\label{eq:bmm_posterior}
    P(\alpha, \bm{\rho} | \bm{R}_1, \ldots, \bm{R}_N ) \propto \frac{\pi(\alpha)\pi(\bm{\rho})}{Z_n(\alpha)^N} \exp \left\{-\frac{\alpha}{n}\sum_{j=1}^N d(\bm{R}_j,\bm{\rho)}\right\}.
\end{equation}
To perform inference, \cite{vitelli2018} proposed a Markov Chain Monte Carlo (MCMC) algorithm based on a Metropolis-Hastings scheme (see \cite{sorensen2019} for details on the implementation). A finite mixture Mallows model has already been introduced in \cite{vitelli2018} to handle the assessor heterogeneity with a model-based approach.

\subsection{Bayesian Mallows model with covariates (BMMx)}\label{sec:model_bmmx}
The Bayesian Mallows model can be modified to account for assessor-related covariates by incorporating the covariate information in the mixture of Mallows models for the ranking data, so that covariates inform the clustering assignment of the assessors. To achieve this, we use a similarity function that favors the clustering of assessors when their covariates are similar, drawing inspiration from the PPMx framework. 

Let $z_1,\ldots, z_N$ be a set of cluster assignment latent variables assigning the assessors into one of $C$ clusters, and we can equivalently define a partition $S$ of the $N$ assessors into $C$ clusters, denoted as $S=\{S_1,\ldots,S_C\}$, where $z_j := c \Leftrightarrow j \in S_c$ for $j=1,\ldots,N$ and $c=1,\ldots,C$. Here the number of clusters $C$ is assumed fixed and known. Furthermore, let $\bm{x}_1,\ldots,\bm{x}_N$ denote the covariate vectors for each of the $N$ assessors, where $\bm{x}_j=(x_{j1},\ldots,x_{jK})$ is a vector containing all $K$ covariates measured for the $j$-th assessor. Covariates can be continuous, categorical, or count-type, and we first assume that the covariates are complete (i.e. no missingness in the covariate data).  

Coherently to \cite{vitelli2018}, the rankings given by the assessors belonging to cluster $c \in \{1, \ldots, C\}$ are assumed to follow a Mallows model with parameters $\alpha_c$ and $\bm\rho_c$. Assuming conditional independence between the $C$ clusters, the augmented data formulation of the likelihood for the observed rankings $\bm{R}_1, \ldots, \bm{R}_N$ is given by
\begin{equation}\label{eq:bmmx_likelihood}
    P(\bm{R}_1, \ldots, \bm{R}_N | \{\bm{\rho}_c, \alpha_c \}_{c=1}^C, z_1,\ldots, z_N) =  \prod_{j=1}^N \frac{1_{\mathcal{P}_n}(\bm{R}_j)}{Z_n(\alpha_{z_j})} \exp \left\{ - \frac{\alpha_{z_j}}{n} d(\bm{R}_j, \bm{\rho}_{z_j}) \right\}.
\end{equation}
To include the covariates $\bm{x}_1,\ldots, \bm{x}_N$, we follow the PPMx framework \cite{muller2011}, where it is assumed that the cluster partitions are a priori distributed according to 
\begin{equation}\label{eq:muller_prior}
    P(S | \bm{x}_1,\ldots,\bm{x}_N; \bm{\tau}) \propto \prod_{c=1}^C g(\mathcal{X}_c) h(S_c, \tau_c)
\end{equation}
where $\bm{\tau} = (\tau_1, \ldots, \tau_C),$ with $\tau_c \geq 0$, $c=1,\ldots,C$ and $\sum_{c=1}^C \tau_c =1$ being the mixture parameters. We define $\mathcal{X}_c=\{ \bm{x}_j, j \in S_c \}$ as the set of covariate vectors $\bm{x}_j$ corresponding to assessors assigned to cluster $c$, for $c=1,\ldots,C$. Here $g(\cdot)$ is a similarity function, which measures the similarity of the covariates of the assessors in a cluster $c$, and $h(\cdot)$ is a cohesion function, which measures how likely it is a priori that the assessors in $S_c$ are in the same cluster. In \cite{muller2011} arguments were provided in support of the fact that a well-defined similarity function would need to fulfill two criteria: 
\begin{enumerate}[(i)]
    \item symmetry with respect to permutations of the sample indices $j$, i.e., $g(\mathcal{X}_c) = g(\bm{x}_{j_1},\ldots,\bm{x}_{j_c}) = g(\bm{x}_{\sigma(j_1)},\ldots,\bm{x}_{\sigma(j_c)})$ with $S_c = \{j_1,\ldots,j_c\}$ and $\sigma(\cdot)$ is a permutation operator of dimension $j_c$.
    \item scales across sample sizes, so that $g(\mathcal{X}_c) = \int g(\mathcal{X}_c, \bm{t})\mathrm{d}\bm{t} = \int g(\bm{x}_{j_1},\ldots,\bm{x}_{j_c}, \bm{t})\mathrm{d}\bm{t}$, i.e., the similarity of any cluster is the similarity of any augmented cluster.
\end{enumerate}
We will discuss these two criteria, and possible alternative choices, when discussing the choice of the similarity function $g(\cdot)$ for BMMx in Section \ref{sec:model_bmmx_simfun}.
We a priori assume that the cluster partitions are distributed according to
\begin{equation}\label{eq:bmmx_cluster_lab_prior}
    P(S| \bm{x}_1, \ldots, \bm{x}_N; \bm{\tau}) \, \propto \,  \prod_{c=1}^C \tau_{c}^{|S_c|} g(\mathcal{X}_c) 
\end{equation}
where, for coherence with BMM without covariates, the cohesion function of the $c$-th cluster is simply defined as $\tau_c^{|S_c|}$, the probability that an assessor belongs to the $c$-th cluster (i.e the mixture parameter, or cluster probability). The cluster probabilities are assigned a standard symmetric Dirichlet prior with concentration parameter $\psi$,
\begin{equation}\label{eq:cluster_prob_prior}
    P(\bm{\tau} | \psi) = \Gamma(\psi C)\Gamma(\psi)^{-C}\prod_{c=1}^C \tau_c^{\psi-1}.
\end{equation}
The function $g(\cdot)$ is the similarity function as defined under the PPMx framework and is further discussed in Section \ref{sec:model_bmmx_simfun}.
For ease of notation, we denote $\bm{R} = \{\bm{R}_1,\ldots,\bm{R}_N\}$ and $\bm{x} = \{\bm{x}_1,\ldots,\bm{x}_N\}$. The full posterior for all parameters of interest is
\begin{align*}
    P(S, \{\bm{\rho}_c, \alpha_c, \tau_c \}_{c=1}^C | \bm{R}, \bm{x}) \, \propto \, P(\bm{R}| \{\bm{\rho}_c, \alpha_c, \tau_c \}_{c=1}^C, S, \bm{x}) P(S, \{\bm{\rho}_c, \alpha_c, \tau_c \}_{c=1}^C | \bm{x}).
\end{align*}
We assume a priori that $\bm{x}$ only influences the cluster partitions $S$, and that $\bm{R}$, $\bm{\rho}$ and $\alpha$ are conditionally independent of $\bm{x}$ and $\tau$ given $S$. We also assume that the cluster assignment probabilities $\tau$ are independent of the covariates $\bm{x}$. In the formulation below, note that we can equivalently define the cluster labels $z_1,\ldots, z_N$ when $S$ is known. The full posterior then becomes:
\begin{equation}\label{eq:bmmx_posterior}
    P(S, \{\bm{\rho}_c, \alpha_c, \tau_c \}_{c=1}^C | \bm{R}, \bm{x}) \, \propto \, 
      P(\bm{R} | S, \{\bm{\rho}_c, \alpha_c \}_{c=1}^C) P(\{\bm{\rho}_c \}_{c=1}^C | S) P(\{\alpha_c \}_{c=1}^C | S) P(S | \bm{\tau}, \bm{x}) P(\bm{\tau})
\end{equation}
where $P(\bm{R} | S, \{\bm{\rho}_c, \alpha_c \}_{c=1}^C)$ is as in equation \eqref{eq:bmmx_likelihood}, $P(\{\bm{\rho}_c \}_{c=1}^C | S ) = \prod_{c=1}^C \pi(\bm{\rho}_c)$ and $P(\{\alpha_c \}_{c=1}^C | S)=\prod_{c=1}^C \pi(\alpha_c)$ with $\pi(\bm{\rho}_c)$ and $\pi(\alpha_c)$ as in equation \eqref{eq:bmm_priors}, $P(S | \{\tau_c \}_{c=1}^C, \bm{x})$ is as in equation \eqref{eq:bmmx_cluster_lab_prior} and $P(\bm{\tau})$ as in equation \eqref{eq:cluster_prob_prior}.

\subsection{Choice of the similarity function $g(\cdot)$}\label{sec:model_bmmx_simfun}

We propose two methods for measuring the similarity function: an augmented model approach, and a goodness-of-fit approach. By presenting both the augmented model and the goodness-of-fit function approach, we offer flexibility in selecting the most suitable method for incorporating covariate information in the clustering process, depending on the specific data and research question at hand.

Let $\mathcal{X}_{ck} = \{x_{jk}, j \in S_c\}$ be the covariate set associated with cluster $c$ for covariate $k$. 
The similarity function used in \eqref{eq:bmmx_cluster_lab_prior} is defined as
\begin{equation}\label{eq:muller_similarity_K}
    g(\mathcal{X}_{c}) = \prod_{k=1}^K g_k(\mathcal{X}_{ck}) 
\end{equation}
where $g_k(\cdot)$ is the similarity function associated with covariate $k$, which depends on the covariate type. The similarity function is a key model component, and it has to be chosen appropriately. 

\subsubsection{Augmented similarity function}

A possible choice for the similarity function is proposed in \cite{muller2011}: the covariates of the assessors assigned to cluster $c$ are random variables distributed according to a shared auxiliary prior $q(x_c |\xi_c)$. This prior is the same for each cluster, except for the cluster-specific hyperparameters $\bm{\xi}_c$, which also have a prior (here denoted by $q_{\xi}$). The prior of the cluster-specific hyperparameters is then integrated out, to obtain the similarity function via an augmented model strategy. In this way, the similarity between assessors is measured in terms of how their covariates fit to a common distribution. The similarity function from \cite{muller2011} takes the form
\begin{equation}\label{eq:muller_similarity}
    g_k(\mathcal{X}_{ck}) = \int \prod_{j \in S_c}q(x_{jk} | \bm{\xi}_{z_j}) q_{\xi}(\bm{\xi}_{z_j}) \mathrm{d} \bm{\xi}_{z_j}.
\end{equation}
It is shown in \cite{muller2011} that \eqref{eq:muller_similarity} satisfies the two criteria (i) and (ii) above. 
The choice of \cite{muller2011} greatly simplifies posterior inference, as it is easier to modify the similarity function with an added distributional factor for each covariate type, compared to constructing a joint model for each covariate in various data formats. We outline a possible strategy for measuring the similarity function via the augmented model for continuous and categorical covariates below.

In the case of the $k$-th covariate being continuous, we suggest using the augmented model in equation \eqref{eq:muller_similarity_K} with $x_{jk} \sim N(\mu_c, \sigma)$ for $j \in S_c$, and a single hyperparameter $\xi_c = \mu_c $ with a normal conjugate hyperprior $\mu_c \sim N(\mu_0, \sigma_0)$. $\sigma$ is fixed to the same value for all clusters as suggested in \cite{muller2011}, but a hyperprior for $\sigma$ could also be used. 
The similarity function is then,
\begin{equation}\label{eq:sim_fun_aug_continuous}
    g_k(\mathcal{X}_{ck}) \, \propto \, \left( \frac{\Tilde{\sigma}}{\sigma}\right)^{|S_c|} \exp{\left\{\frac{1}{2}  \sum_{j \in S_c} \left( \frac{\Tilde{\mu_j}^2}{\Tilde{\sigma}^2} - \frac{x_{jk}^2}{\sigma^2} \right) \right\} }
\end{equation}
using the normal-normal conjugacy (see \cref{supp:sim_fun_cont} for details on the derivation of equation \eqref{eq:sim_fun_aug_continuous}) where $\Tilde{\mu}_j = \sigma_0^2x_{jk}/(\sigma^2 + \sigma_0^2) + \sigma^2 \mu_0 /(\sigma^2 + \sigma_0^2)$ and $\Tilde{\sigma} = (\sigma^{-2}+ \sigma_0^{-2})^{-1}$.
To fix reasonable values for $\sigma$, $\mu_0$, and $\sigma_0$, we have set $\sigma = c_1 \hat{s}_x$, $\sigma_0= c_2 \hat{s}_x$, $\mu_0 = \Bar{x}$, where $\hat{s}_x$ is the empirical variance of the covariate and $\Bar{x}$ is the mean of the covariate. $c_1$ and $c_2$ need to be set based on the data: $c_1$ is a scaling factor for the variance of the augmented distribution of the covariate, and $c_2$ should be reasonably large (in the context of the data at hand). In Section \ref{sec:simulations} we set up a simulation study to explore possible strategies to tune these hyperparameters. 

If the $k$-th covariate is instead categorical with $B$ categories, a default choice is based on a Multinomial prior. We assume $x_{jk} \sim \mathcal{M}(\nu_{c,1}, \ldots, \nu_{c,B})$, for $j \in S_c$ and with $B$ categories, so that $\bm{\nu}_c=(\nu_{c,1},\ldots,\nu_{c,B})$ is the vector of hyperparameters defining the probabilities of the covariate categories as $P(x_{jk}=b) = \nu_{c,b}$. Analogously to the continuous covariate case, we can assume the conjugate Dirichlet hyperprior $q_{\xi}$ so that $\bm{\nu}_c \sim \mathcal{D}(\varphi_1,\ldots,\varphi_B)$. Then the similarity function is
\begin{equation}\label{eq:sim_fun_aug_categorical}
    g_k(\mathcal{X}_{ck}) \, \propto \, \prod_{b=1}^B \frac{ \Gamma((\sum_{j \in S_c} 1_{x_{jk}=b}) + \varphi_b) }{ \Gamma( \sum_{b=1}^B [(\sum_{j \in S_c} 1_{x_{jk}=b}) + \varphi_b]) }
\end{equation}
which is a Dirichlet-Multionomial model without the multinomial coefficient (see \cref{supp:sim_fun_categorical} for a detailed derivation of equation \eqref{eq:sim_fun_aug_categorical}). For sake of simplicity, we assume $\varphi_1=\ldots=\varphi_B=\varphi$ (i.e. symmetric Dirichlet distribution), and show a possible strategy to tune this hyperparameter in Section \ref{sec:simulations}.

\subsubsection{Goodness-of-fit similarity function}
The second proposed approach for measuring the similarity between covariates is instead based on a deterministic goodness-of-fit function, and thus has a quite different rationale from the augmented model proposed in \cite{muller2011}. In this approach, the goodness-of-fit function measures the distance between the observed covariate values and the average covariate values of the assessors in a given cluster. This approach provides a simple yet effective method for incorporating covariate information into the clustering process, and can be particularly useful in cases where the augmented model may not be appropriate due to the data-generating process associated to the covariate. Nonetheless, this is a completely nonparametric approach, which is not limited by any distributional assumption. 

This goodness-of-fit similarity function takes the following form
\begin{equation}\label{eq:sim_fun_alternative}
    g_k(\mathcal{X}_{ck}) = \frac{1}{|S_c|} \sum_{j \in S_c} \frac{\delta(x_{jk},\Bar{x}_{ck})}{\sum_{l=1}^C \delta(x_{jk},\Bar{x}_{lk})} ,
\end{equation}
where $\Bar{x}_{ck}$ is the $k$-th covariate cluster centroid, and $\delta(\cdot, \cdot)$ is a similarity function measuring the ``closeness'' of $x_{jk}$ to the covariate cluster centroid. The function $g(\cdot)$ is in other words measuring the covariates' goodness-of-fit to the cluster. There are several options for choosing $\delta(\cdot, \cdot)$ and $\Bar{x}_{ck}$, and the choice will depend on the type of covariate, and the data application. 

As evident from \eqref{eq:sim_fun_alternative}, this similarity is completely deterministic, and will induce an improper prior on the cluster partition. Therefore, while criterion (i) introduced above for a similarity function being well-defined is still sensible in this case, criterion (ii) makes no particular sense, since $g(\cdot)$ is not a probability measure. Other criteria might be more suitable here, like the following invariance criterion with respect to \emph{nuisance covariates} (i.e., covariates that show no variation). Let us define a nuisance covariate such that $\mathcal{X}^\star:=\{x^\star_j=x^\star\ \ \forall j=1,\ldots,N\}$. Then,
\begin{itemize}
    \item[(iii)] the similarity function is \emph{invariant} if $g(\mathcal{X}_c \cup \mathcal{X}^\star) \equiv a \cdot g(\mathcal{X}_c)$ for some constant value $a$ independent on the covariate value.
\end{itemize}

\begin{proposition}
    The similarity function \eqref{eq:sim_fun_alternative} satisfies the two criteria (i) and (iii) above. 
\end{proposition}
\begin{proof}
    Criterium (i) of symmetry holds as the covariates enter a sum, which is a symmetric operator with respect to its arguments. Criterium (iii) of covariate invariance is a direct consequence of the similarity definition in \eqref{eq:sim_fun_alternative}, with $a = 1/C.$ 
\end{proof}

In the case of a continuous covariate $k$, we propose and use throughout this paper
\begin{equation}\label{eq:sim_fun_alternative_cont}
    g_k(\mathcal{X}_{ck}) = \frac{1}{|S_c|} \sum_{j \in S_c} \frac{(1+\theta| x_{jk}-\Bar{x}_{ck}|)^{-1}}{\sum_{l=1}^C (1+\theta| x_{jk}-\Bar{x}_{lk}|)^{-1}} ,
\end{equation}
where $\theta$ is a hyperparameter and $\Bar{x}_{ck}$ is the average over the $k$-th covariate of the assessors in cluster $c$. 
In the case of a categorical covariate $k$, we propose 
\begin{equation}\label{eq:sim_fun_alternative_cat}
    g_k(\mathcal{X}_{ck}) = \frac{1}{|S_c|} \sum_{j \in S_C} \frac{1 + \gamma 1_{x_{jk} = \Bar{x}_{ck}}}{\sum_{l=1}^C 1 + \gamma 1_{x_{jk} = \Bar{x}_{lk}}} ,
\end{equation}
where $\gamma$ is a hyperparameter, and the covariate cluster centroid  $\Bar{x}_{ck}$ is the mode of the $k$-th covariate of the assessors in cluster $c$. 

The hyperparameters $\theta$ and $\gamma$ determine how much weight is given to the covariates as compared to the data in the mixture model, and therefore during the clustering estimation process. Larger values of $\theta$ and $\gamma$ indicate a higher degree of importance assigned to the covariates in estimating the clusters, while lower values give more importance to the data. We demonstrate how to tune these hyperparameters in a sensitivity study in Section \ref{sec:simulations}.

\begin{remark}
    Note that up until now, the prior for the cluster partitions is defined via a product over the partitions $S_1,\ldots,S_C$, as in \eqref{eq:muller_prior}, while for the inference in BMMx we will instead define it as a product over the cluster labels $z_1,\ldots, z_N$. These are equivalent, as there is a one-to-one correspondence between the cluster labels and the cluster partitions.
\end{remark}

\section{Inference in the BMMx via MCMC}\label{sec:mcmc}

Since the posterior distribution in equation \eqref{eq:bmmx_posterior} is intractable, to sample from it we build a Metropolis-within-Gibbs MCMC algorithm iterating between three steps: a Gibbs step where we update the cluster assignment probabilities, a Metropolis-Hastings step where we update the cluster-wise consensus ranking and scale parameter, and another Gibbs step where we update the cluster label assignments.

In the first step of the algorithm, the cluster assignment probabilities $\tau_1,\ldots,\tau_C$ are updated in a Gibbs sampler step. Since $\bm{\tau} \sim \mathcal{D}(\psi)$ a priori, which is conjugate to the multinomial likelihood, and since $\bm{\tau}$ only depends on $\bm{z}$ a posteriori,  the full conditional distribution takes the form $\bm{\tau}|\bm{z} \sim \mathcal{D}(\psi + n_1,\ldots,\psi+n_C)$, $n_c=\sum_{j=1}^N 1_c(z_j)$ for $c=1,\ldots,C$. 

The second step of the algorithm updates the consensus rankings $\bm{\rho}_1,\ldots,\bm{\rho}_C$ and the scale parameters $\alpha_1,\ldots,\alpha_C$ in two separate Metropolis-Hastings steps. For each $c=1,\ldots,C$, a new consensus ranking $\bm{\rho}'_c \in \mathcal{P}_n$ is proposed using the ``leap-and-shift'' proposal distribution described in \cite{vitelli2018}. The acceptance probability for updating $\bm{\rho}_c$ in the MCMC algorithm is
\begin{equation}\label{eq:accept_prob_rho}
    \min\left\{ 1, \frac{P_l(\bm{\rho}_c|\bm{\rho}_c')}{P_l(\bm{\rho}_c'|\bm{\rho}_c)}\exp \left[ -\frac{\alpha_c}{n} \left( \sum_{j:z_j=c} d(\mathbf{R}_j,\bm{\rho}_c') - \sum_{j:z_j=c} d(\mathbf{R}_j,\bm{\rho}_c)\right)\right] \right\}.
\end{equation}
In equation (\ref{eq:accept_prob_rho}) above, $P_l$ denotes the probability mass function associated to moving from one $n$-dimensional ranking vector to another according to the leap-and-shift proposal. The parameter $l$ denotes the number of items perturbed in the consensus ranking $\bm{\rho}_c$ to get a new proposal $\bm{\rho}_c'$, and is used to tune the acceptance probability. $d(\cdot, \cdot)$ is the distance measure introduced in Section \ref{sec:model_bmm}. 
For the scale parameters $\alpha_1,\ldots,\alpha_c$, for each $c=1,\ldots,C$ we sample $\alpha_c'$ from a log-normal distribution $\log \mathcal{N}(\log(\alpha_c), \sigma_\alpha^2)$. The proposal is accepted with probability
\begin{equation}\label{eq:accept_prob_alpha}
    \min\left\{ 1, \frac{Z_n(\alpha_c)^N \pi(\alpha_c')\alpha_c'}{Z_n(\alpha'_c)^N \pi(\alpha_c)\alpha_c}\exp \left[ -\frac{(\alpha_c' - \alpha_c)}{n}  \sum_{j:z_j=c} d(\mathbf{R}_j,\bm{\rho}_c) \right] \right\}
\end{equation}
where $\sigma_\alpha^2$ is used to tune the acceptance probability. An additional parameter, denoted as $\alpha_{\textnormal{jump}}$, can also be used to update $\alpha_c$ every $\alpha_{\textnormal{jump}}$ updates of $\bm{\rho}_c$ to promote better mixing of the MCMC and to avoid unnecessary (and time-consuming) updates of a single scale parameter as compared to the multivariate updating of the consensus ranking.

The cluster label assignments $z_1, \ldots, z_N$ are updated in the third and final step of the algorithm. This is done in a Gibbs sampler step, where for each $j=1,\ldots,N$ we sample from the full conditional distribution
\begin{equation} \label{eq:gibbs_clusterlabels}
    P(z_j = c | z_{-j}, \bm{\rho}_c, \alpha_c, \tau_c, \bm{R}_j, \{\mathcal{X}_{ck}\}_{k=1}^K) \, \propto \, \frac{1}{Z_n(\alpha_{c})} \exp\left\{-\frac{\alpha_c}{n} d(\bm{R}_j, \bm{\rho}_{c})\right\} \tau_c \prod_{k=1}^K g_k(\mathcal{X}_{ck})
\end{equation}
where $z_{-j}$ are all labels $\bm{z}$ excluding $z_j$, and $\mathcal{X}_{ck} = (x_{lk}, l : \{z_l = c \} \cup \{j\}, l= 1,\ldots, N)$. See the detailed derivation of the full conditional distribution for $z_j$ in \cref{supp:gibbs_sampling}.

The MCMC algorithm described above is summarized in Algorithm \ref{alg:bmmx_mcmc}. Note that \emph{aug} is a boolean that activates the use of the augmented similarity function, and $M$ is the number of MCMC iterations. 

\begin{algorithm}[!htb]
\SetAlgoLined
\scriptsize
\textbf{input:} $\mathbf{R}_1,\ldots,\mathbf{R}_N$,  $\mathbf{x}_1,\ldots,\mathbf{x}_N$; $C$, $\psi$, $c_1$, $c_2$, $\varphi$, $\theta$, $\gamma$, $\sigma_\alpha^2$, $\alpha_{\textnormal{jump}}$, $d(\cdot, \cdot)$, $l$, $Z_n(\alpha)$, $M$, $aug$ \\
\textbf{output:} posterior distributions of $\bm{\rho}_1,\ldots,\bm{\rho}_C$, $\alpha_1,\ldots,\alpha_C$, $\tau_1,\ldots,\tau_C$, $z_1,\ldots,z_N$\\
 \textbf{initialization:} randomly generate $\bm{\rho}_{1,0},\ldots,\bm{\rho}_{C,0}$, $\alpha_{1,0},\ldots,\alpha_{C,0}$, $\tau_{1,0},\ldots,\tau_{C,0}$ and $z_{1,0},\ldots,z_{N,0}$ \\
\For{$m \gets 1$ \KwTo $M$}{
    \textbf{Gibbs step}: update $\tau_1,\ldots,\tau_C$ \\
    compute: $n_c=\sum_{j=1}^N 1_c(z_{j,m-1})$ for $c=1,\ldots,C$ \\
    sample: $\tau_1,\ldots,\tau_C \sim \mathcal{D}(\psi + n_1,\ldots,\psi+n_C)$ \\
    \For{$c \gets 1$ \KwTo $C$}{
        \textbf{M-H step}: update $\bm{\rho}_c$ \\
        sample: $\bm{\rho}_c' \sim \textnormal{LS}(\bm{\rho}_{c,m-1}, l)$ and $u \sim \mathcal{U}(0,1)$ \\
        compute: \textit{ratio} $\gets$ equation \eqref{eq:accept_prob_rho} with $\bm{\rho}_c \gets \bm{\rho}_{c,m-1}$, $\alpha_c \gets \alpha_{c,m-1}$ and $z_j \gets z_{j,m-1}$\\
        \eIf{$u <$ ratio}{$\bm{\rho}_{c,m} \gets \bm{\rho}'_{c}$}{$\bm{\rho}_{c,m} \gets \bm{\rho}_{c,m-1}$}\leavevmode
        \textbf{M-H step}: update $\alpha_c$ \\
        sample: $\alpha_c' \sim \log \mathcal{N}(\log(\alpha_{c,m-1}), \sigma_\alpha^2)$ and $u \sim \mathcal{U}(0,1)$ \\
        compute: \textit{ratio} $\gets $ equation \eqref{eq:accept_prob_alpha} with $\bm{\rho}_c \gets \bm{\rho}_{c,m}$, $\alpha_c \gets \alpha_{c,m-1}$ and  $z_j \gets z_{j,m-1}$\\
        \eIf{$u <$ ratio}{$\alpha_{c,m} \gets \alpha'_{c}$}{$\alpha_{c,m} \gets \alpha_{c,m-1}$}
    }\leavevmode
    \textbf{Gibbs step}: update $z_1,\ldots,z_N$ \leavevmode \\
    \For{$j \gets 1$ \KwTo $N$}{
    \For{$c \gets 1$ \KwTo $C$}{
        \eIf{aug}{
            \lIf{$x$ continuous}{$g_k(\mathcal{X}_{ck})$ $\gets$ equation \eqref{eq:sim_fun_aug_continuous}}
            \lIf{$x$ categorical}{$g_k(\mathcal{X}_{ck})$ $\gets$ equation \eqref{eq:sim_fun_aug_categorical}}
            }{
            \lIf{$x$ continuous}{$g_k(\mathcal{X}_{ck})$ $\gets$ equation \eqref{eq:sim_fun_alternative_cont}}
            \lIf{$x$ categorical}{$g_k(\mathcal{X}_{ck})$ $\gets$ equation \eqref{eq:sim_fun_alternative_cat}}
            }
    		compute: $p_{cj} \gets$ equation \eqref{eq:gibbs_clusterlabels} with $\tau_c \gets \tau_{c,m}$, $\alpha_c \gets \alpha_{c,m}$ and $\bm{\rho}_c \gets \bm{\rho}_{c,m}$}
    		sample: $z_{j,m} \sim \mathcal{M}(p_{1j},\ldots,p_{Cj})$
    	}
}
\caption{MCMC scheme for inference in BMMx} \label{alg:bmmx_mcmc}
\end{algorithm}

\subsection{Incomplete data}\label{sec:model_mcmc_missing}
It is common in real-world applications that data are incomplete, and a high level of missingness might pose challenges to the analysis. In the case of BMMx, missing data can occur in both the ranking data $\bm{R}$ and the covariates $\bm{x}$. For instance, in the case of ranking data, it is not unusual for assessors to rank only a subset of items: ranks can be missing at random, the assessors may only have ranked the, in their opinion, top-$k$ items, or they can be presented with only a subset of items that they rank. Similarly, in the covariates, missing values can arise due to a variety of reasons, such as measurement errors, data entry errors, or non-response. Fortunately, missing data situations can be handled easily in the Bayesian framework, by adding a data augmentation step to the MCMC algorithm. 

Suppose that each assessor $j$ has ranked a subset of items $\mathcal{A}_j \subset \{A_1,\ldots,A_n\}$ consisting of $n_j \leq n$ items. The observed rank vector $\bm{R}_j,$ for $j=1,\ldots,N,$ will then include only the rankings of the items in $\mathcal{A}_j$ (and not necessarily the top-$n_j$ ones). We thus define the augmented data vectors $\Tilde{\bm{R}}_1,\ldots,\Tilde{\bm{R}}_N,$ where the ranks of the non-ranked items are latent variables,  whose posterior is estimated via data augmentation in the MCMC algorithm. This MCMC algorithm alternates between two steps: (i) sampling the augmented ranks given the current values of all other parameters in the model, and (ii) sampling all model parameters given the current values of the augmented ranks. Let $\mathcal{S}_j = \{\Tilde{\bm{R}}_j \in \mathcal{P}_n : \forall A_{i_1}, A_{i_2} \in \mathcal{A}_j \textnormal{ s.t. } R_{{i_1}j} < R_{{i_2}j}  \Rightarrow \Tilde{R}_{i_1 j} < \Tilde{R}_{i_2 j}\}$ be the set of possible augmented random vectors that comply with the original partially ranked items. This practically means that in $\Tilde{\bm{R}}_j = (\Tilde{R}_{1j},\ldots, \Tilde{R}_{nj})$ the respective ordering between items in $\mathcal{A}_j$ must be the same as in $\bm{R_j}$, whereas the augmented ranks of the items in $\mathcal{A}_j^c$ are unconstrained. 

First, given the augmented ranks, all model parameters are sampled from the posterior in \eqref{eq:bmmx_posterior} $P(\bm{z}, \{\bm{\rho}_c, \alpha_c, \tau_c \}_{c=1}^C |  \Tilde{\bm{R}}_1,\ldots,\Tilde{\bm{R}}_N, \bm{x})$ as explained in Section \ref{sec:mcmc} above. Then, given $\bm{z}$ and $\{\bm{\rho}_c, \alpha_c, \tau_c \}_{c=1}^C$, each $\Tilde{\bm{R}}_j$ is updated separately (as the augmented rank vectors are conditionally independent from each other given the cluster label associated to the assessors and the other model parameters). A proposed $\Tilde{\bm{R}}_j'$ is sampled uniformly in $\mathcal{S}_j$, so as to respect the mutual orderings observed in the data. The proposed $\Tilde{\bm{R}}_j'$ is then accepted with probability 
\begin{equation}\label{eq:accept_data_aumentation}
    \min \left\{1, \exp \left[ -\frac{\alpha_{z_j}}{n} \left( d(\Tilde{\bm{R}}_j', \bm{\rho}_{z_j}) - d(\Tilde{\bm{R}}_j, \bm{\rho}_{z_j})\right) \right] \right\}.
\end{equation}
The MCMC algorithm for the missing data case is described in Algorithm \ref{alg:bmmx_mcmc_missing} in \cref{supp:bmmx_mcmc_missing}.

In the case of incomplete covariates, one can simply ignore the missing covariate values in the computation of the similarity function. Since the covariates (via the similarity function) are only considered as prior information for the cluster labels, ignoring such missing values is equivalent to assuming no prior information from the covariates on the corresponding assessors, with a multiplying factor that only involves the cohesion function. Therefore, missing values in the covariates are automatically handled in the model by using the original BMM prior.

\section{Simulation experiments}\label{sec:simulations}
In this section, we provide a comprehensive evaluation of the performance of the Bayesian Mallows model with covariates (BMMx) and compare it against the Bayesian Mallows model (BMM). We assess the two methods based on their ability to correctly assign cluster labels to assessors, and calculate the mean posterior probability for each assessor being assigned to the correct cluster. We also conduct a sensitivity analysis to investigate the impact of different hyperparameters included in the similarity functions. 
Through these evaluations, we aim to provide a deeper understanding of the strengths and weaknesses of the BMMx and BMM methods in various scenarios.

Two datasets are needed for the simulation experiments: the rankings $\bm{R}$ and the covariates $\bm{x}$. Let us first focus on the rankings. The generation of the dataset $\bm{R}$ was done in the following way: the observed complete ranking of the $j$th assessor was sampled from a Mallows model with a fixed $\alpha$ (same for all clusters), and with a cluster-specific $\bm{\rho}_c,$ with $c=1,\ldots,C$ and $C$ fixed. Specifically, $R_j \sim \textnormal{Mallows}(\bm{\rho}_{w_j}, \alpha)$ for $j=1,\ldots,N$, and with $w_j\in\{1,\ldots,C\}$ being the true cluster label assignment. To generate the consensus rankings for the different clusters, $\bm{\rho}_c$ with $c=1,\ldots,C$, a swapping procedure was followed. The swapping procedure consisted in swapping the ranks of $s$ items at a certain footrule distance from a ``baseline'' consensus ranking ($\bm{\rho}_1,$ kept fixed) to generate $\bm{\rho}_c$ for $c=2,\ldots,C$. 
More specifically, the swapping was done as follows. We start by setting $\bm{\rho}_1 = (1,2,\ldots,n)$, and initializing $\bm{\rho}_c=\bm{\rho}_1$ for all $c=2,\ldots,C$. Two main steps were then performed. First, a set of items $\mathcal{S}_0$, $|\mathcal{S}_0|=s(C-1)$, were sampled from $\bm{\rho}_1$. Secondly, to generate a set of ``swapped'' items $\mathcal{S}_{\textnormal{swap}}$ which should all have a rank at distance $d_x$ from the ranks of the original items, we randomly added or subtracted $d_\rho$ to the ranks of all items in $\mathcal{S}_0$ to get $\mathcal{S}_{\textnormal{swap}}$. These two steps were repeated until the items in $\mathcal{S}_{\textnormal{swap}}$ fulfilled two criteria: (i) $rank(\mathcal{S}_{\textnormal{swap}}) \subset \{1,2,\ldots,n\}$, and (ii) $\mathcal{S}_{\textnormal{swap}} \notin \mathcal{S}_0$ (no items are swapped back and forth). Once fulfilled, $s$ items from $\mathcal{S}_0$ were swapped with $s$ items from $\mathcal{S}_{\textnormal{swap}}$ to get $\bm{\rho}_c$ for each $c=2,\ldots,C$.

The swapping procedure for generating the consensus rankings was designed to have a measure of ``difficulty'' of the generated clustered data: a larger $d_\rho$ (and $s$) indicates a ranking dataset where the cluster centroids are further apart, and thus easier to estimate. Conversely, a smaller $d_\rho$ indicates more similar centroids and thus a more challenging clustering estimation task.

Analogously, the covariates dataset $\bm{x}$ was generated according to a size parameter $d_x,$ which affected the respective distances between cluster means. Specifically, for $j=1,\ldots,N,$ $x_j \sim N(\mu_{v_j}, \sigma)$, with $v_j\in\{1,\ldots,C\}$ defining the covariate cluster assignment to $C$ groups, and $\mu_c = d_x(c-1)$, for $c=1,\ldots,C$. A larger $d_x$ indicates a more informative covariate dataset, as the covariates are sampled from $C$ well-separated distributions with means $\mu_c$'s further apart from one another, which should thus be more informative in the clustering estimation. On the other hand, a smaller $d_x$ indicates a less informative covariate dataset with less separated cluster-wise covariate distributions.

We computed three performance measures in order to evaluate and compare the methods. Recall that $\bm{w}=\{w_1,\ldots, w_N\}$ is the true cluster label assignment, and let $\bm{\hat{z}}=\{\hat{z}_1,\ldots, \hat{z}_N\}$ be its estimate via the MAP of $\bm{z}$. We consider the following measures of performance:
\begin{enumerate}
    \itemsep0em
    \item the proportion of assessors assigned to the correct cluster $\hat{p} = \frac{1}{N} \sum_{j=1}^N 1_{\hat{z}_j=w_j}$. 
    \item the mean posterior probability of all assessors being assigned to the correct cluster \\
    $z_{\textnormal{post}} = \frac{1}{N} \sum_{j=1}^N P(z_j=w_j |\{\bm{\rho}_c, \alpha_c, \tau_c \}_{c=1}^C,  \bm{R}, \bm{x})$.
\end{enumerate}
\subsection{$K=3$ continuous covariates}
We consider a simulation example with $n=20$ items, and $N=90$ assessors assumed to be grouped into $C=3$ clusters of equal size ($N_c = 30$ assessors in each cluster $c=1,2,3$). We here consider $K=3$ continuous covariates. The two datasets (rankings and covariates) were simulated according to the procedure explained above, with $s=3$ items swapped, and $\alpha = 5$ for all clusters (for the generation of the rankings), and with a fixed variance of $\sigma=1$ for all covariate distributions for all clusters (for the generation of the covariates).

The analysis was carried out with BMM as well as with BMMx (and its similarity function variations) over a grid of values for $d_\rho$ and $d_x$. The aim of these initial experiments was to verify that the method works as expected, and to study how adding the covariates into the analysis affects the final estimation of the clusters. We ran the two methods 10 times for each set of $d_\rho$ and $d_x$, $M=10^4$ MCMC iterations, and burn-in set to $1000$. Figure \ref{fig:sim_p_hat_scatter_3cont} displays the proportion of correctly assigned cluster labels (denoted as $\hat{p}$), over 10 runs with the BMM and BMMx in the experiment described above with $s=3$ in the upper two panels and with $s=4$ in the lower two panels. The results showed that BMMx consistently performed better than or equally to BMM. Additionally, there is a noticeable difference in performance between the two similarity measures. Specifically, our proposed goodness-of-fit similarity function performs better than the augmented similarity function for higher values of $d_\rho$. This suggests that the augmented model may not benefit from more informative covariates, as there is very little improvement in performance as $d_x$ increases. In contrast, when using the goodness-of-fit similarity the model shows a clear improvement in performance with increasing $d_x$. We also see that there is a slightly higher degree of uncertainty in the estimates for the model using the augmented similarity as compared to the other. When increasing $s$ (the number of items swapped), we see that there is an overall increase in performance for both methods, as expected.

These experiments were performed without any tuning of the various associated hyperparameters, which suggests that the ``baseline performance'' of the goodness-of-fit similarity is superior to the augmented similarity: this can be explained by the fact that the augmented similarity measure needs three hyperparameters, compared to the single hyperparameter needed for the goodness-of-fit similarity measure. We show in the remainder of the section how the performance of the two similarity measures changes with varying hyperparameters. 

\begin{figure}[!htb]
\minipage{\textwidth}
  \includegraphics[width=\linewidth]{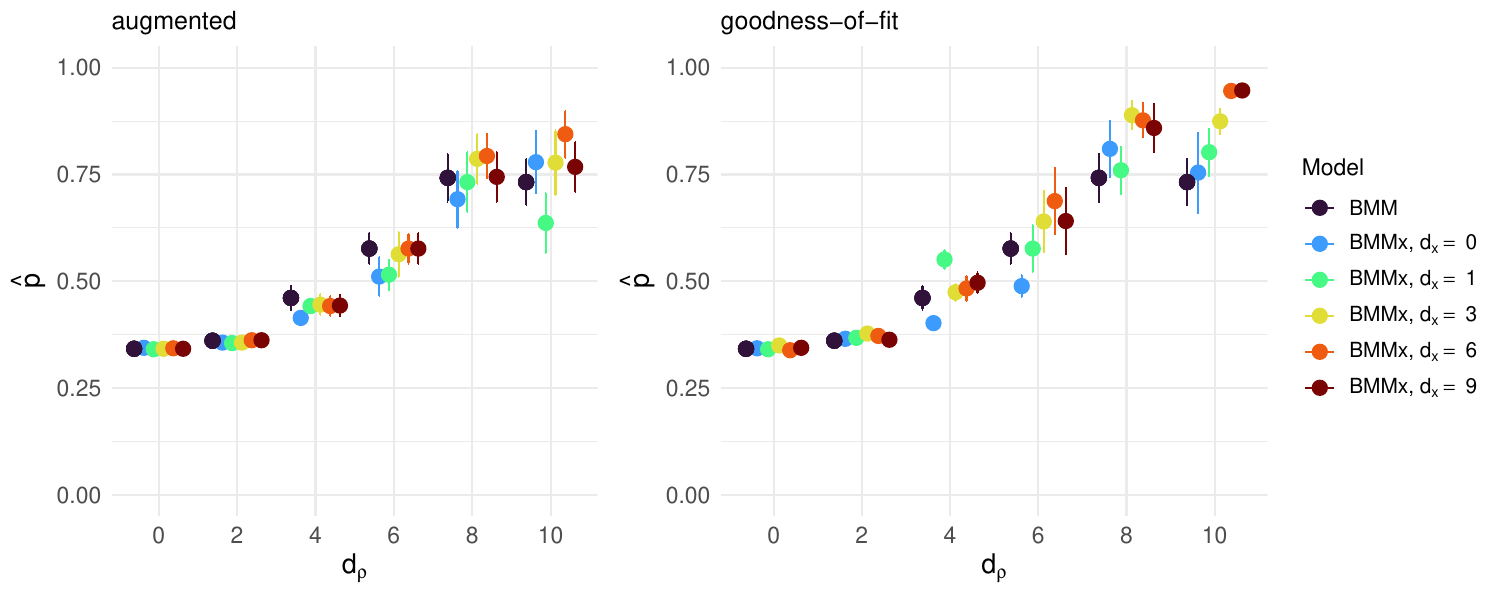}
\endminipage\hfill
\minipage{\textwidth}
  \includegraphics[width=\linewidth]{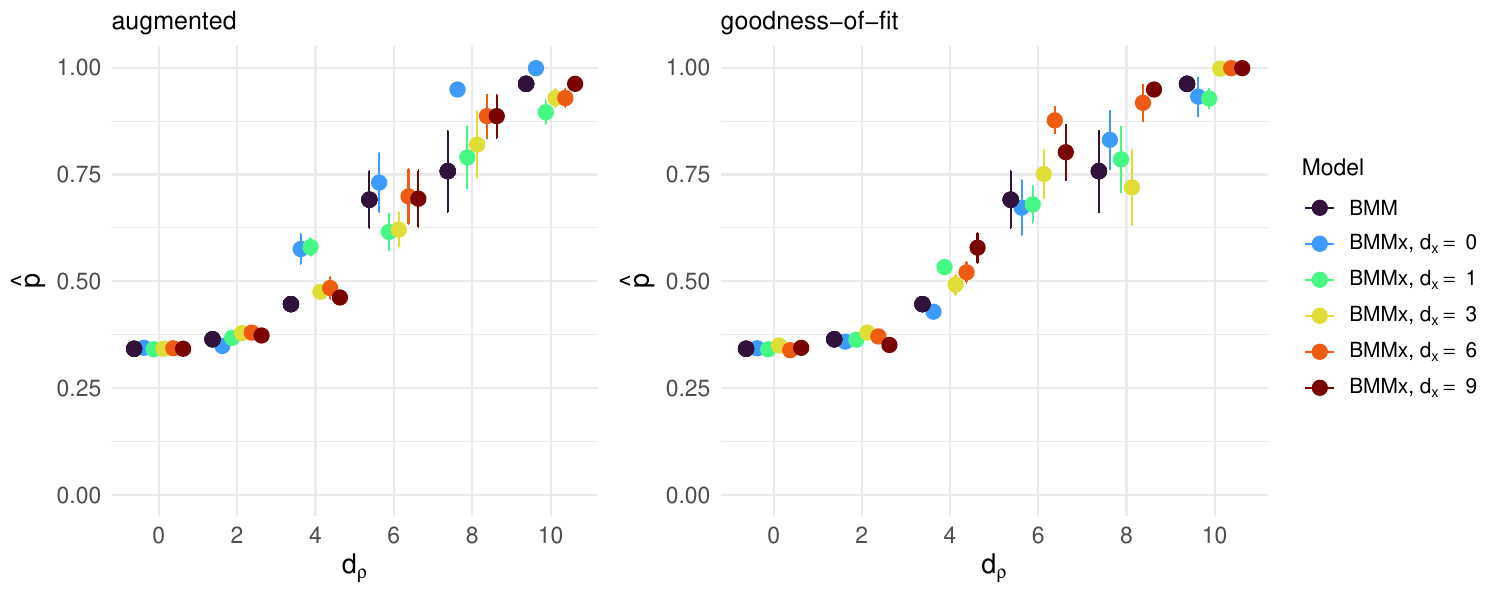}
\endminipage 
\caption{Scatterplot of $\hat{p}$ (y-axis) with respect to $d_\rho$ (x-axis), with variance bars over 10 runs, for BMM (black) and BMMx using covariates with varying $d_x$ (colors, as shown in legend), with $s=3$ (upper panels) and $s=4$ (lower panels), for the augmented similarity function (left) and the goodness-of-fit similarity function (right). Hyperparameters for BMMx: $c_1 = 0.5$, $c_2=10$, $\mu_0 = 1$ (augmented similarity function) and $\theta=1$ (goodness-of-fit similarity function).}
\label{fig:sim_p_hat_scatter_3cont}
\end{figure}

A sensitivity study was then carried out to verify the effect of tuning the hyperparameters introduced in the two similarity functions. We first focused on studying the effect of $c_1$ on the model when using the augmented similarity function, and of $\theta$ when using the goodness-of-fit similarity function, in the case of continuous covariates. The data was generated according to the data-generating process described at the beginning of the section, and the models were tested over a grid of values for $c_1$ and $\theta$ that allowed exploring the performance well (both in terms of range and order of magnitude). Results were collected after running the method on 10 simulated datasets, using $M=10^4$ MCMC iterations and burn-in set to $10^3$. 

Figure \ref{fig:sim_tuning_c1} displays $z_{\textnormal{post}}$ for varying values of $d_x$ and $c_1$ with the augmented similarity function for $d_\rho=6$ (left panel) and $d_\rho=10$ (right panel). The results showed that a higher value for $c_1$ seems to be the optimal choice in the more difficult clustering case ($d_\rho=6$), while in the easier clustering case there seems to be no clear tuning strategy, although $c_1=10$ is a better choice in many of the cases (but not always). A higher $c_1$ corresponds to a less informative auxiliary prior for the augmented distribution of the covariates, suggesting that in this model being less specific is generally better, regardless of how informative the covariate dataset is. It is worth noting that BMM performs as well as or better than BMMx in most of the cases in the current setup. It is possible that a more exhaustive sensitivity study tailored to studying the tuning of $c_1$ would help to better understand the performance of this version of the BMMx model, however we consider this to be out of the scope of this paper, and conclude that the use of the alternative goodness-of-fit similarity function is advisable in general.

\begin{figure}[!htb]
\minipage{\textwidth}
\centering
  \includegraphics[width=\linewidth]{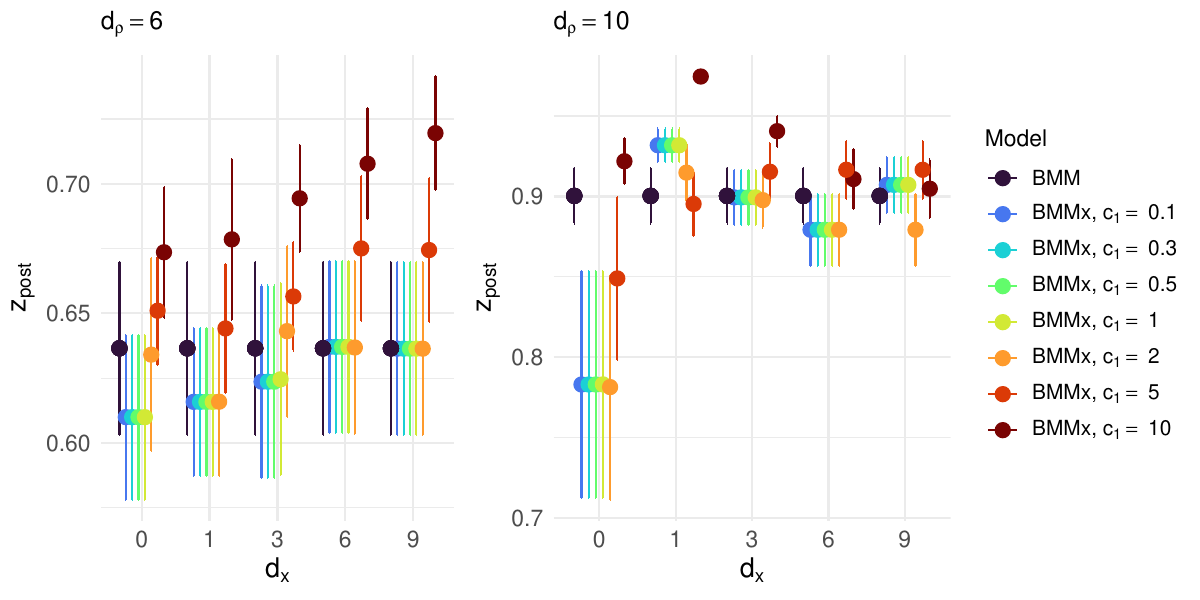}
\endminipage 
\caption{Scatterplot of $z_{\textnormal{post}}$ (y-axis) with respect to $d_x$ (x-axis), with variance bars over 10 runs, for BMM (black) and BMMx with varying $c_1$ (colors, as shown in legend), with the augmented similarity function. $d_{\rho} = 6$ (left) and $d_{\rho} = 10$ (right), $s=3$.}
\label{fig:sim_tuning_c1}
\end{figure}

Similarly, Figure \ref{fig:sim_tuning_theta} displays $z_{\textnormal{post}}$ for varying values of $d_x$ and $\theta$ with the goodness-of-fit similarity function for $d_\rho=6$ (left panel) and $d_\rho=10$ (right panel). Here the results showed an indication that mid-range values of $\theta$ were performing best on average. For higher values of $d_x$ (more informative covariates), there is an indication that higher values of $\theta$ would be better, while for lower values of $d_x$ (less informative covariates), mid-range values of $\theta$ seem to be optimal. This suggests that the hyperparameter is attempting to balance the information content from covariates with the actual data. Indeed, as one would intuitively expect, when the covariates are less informative, making their contribution weigh less (i.e. mid-range $\theta$-values) is optimal, while conversely when the covariates are more informative, weighting them more (i.e. higher $\theta$-values) is preferred. As a tuning strategy, we would therefore suggest to set $\theta$ reasonably large. Overall, BMMx with the proposed goodness-of-fit similarity function always performs as well as or better than BMM, which also shows the superiority of this choice as compared to the augmented similarity.

\begin{figure}[!htb]
\minipage{\textwidth}
\centering
  \includegraphics[width=\linewidth]{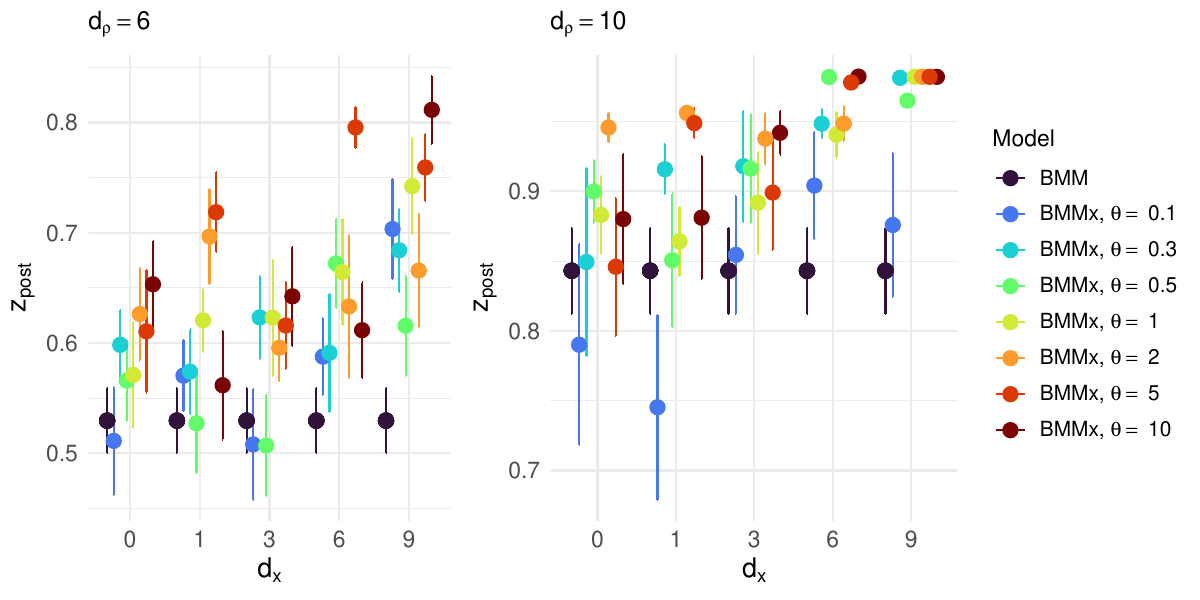}
\endminipage 
\caption{Scatterplot of $z_{\textnormal{post}}$ (y-axis) with respect to $d_x$ (x-axis), with variance bars over 10 runs, for BMM (black) and BMMx with varying $\theta$ (colors, as shown in legend), with the goodness-of-fit similarity function. $d_{\rho} = 6$ (left) and $d_{\rho} = 10$ (right), $s=3$.}
\label{fig:sim_tuning_theta}
\end{figure}


\subsection{$K=1$ categorical covariate}

We now consider only $K=1$ categorical covariate, generated in two scenarios: one where $\bm{x}$ is distinct categories, and one where $\bm{x}$ is randomly distributed categories (essentially just noise). In the first scenario, the covariate dataset was generated so that the covariate describes a clear subgrouping of the data into $B=3$ categories, which are coherent to the true clustering of the assessors. In the second scenario, the covariate dataset was generated from a random sample from a categorical distribution with $B=2$ categories. Similarly to the experiments done for the continuous covariates, we here run the BMM and BMMx over a grid of values for $d_\rho$ and $d_x$, and for BMMx we also varied the hyperparameters $\varphi$ (augmented similarity) and $\gamma$ (goodness-of-fit similarity). Figure \ref{fig:sim_tuning_gamma} displays scatterplots of $\hat{p}$ for the two data cases while also varying $\gamma$. The results from the relatively small tuning study show mainly two things: first, in the case of the goodness-of-fit similarity measure, increasing $\gamma$ seems to improve performance when the covariate is informative, while it does the opposite for the uninformative covariate. When $d_x$ is mid-sized (i.e. 4-6), results are quite similar for the two covariate cases. Moreover, BMMx and BMM perform quite comparably. Secondly, in the case of the augmented model, BMMx performs significantly worse than BMM, especially in the case of the uninformative case (ref. Figure \ref{fig:sim_tuning_psi}). In this case, there is no evident tuning strategy, except to possibly keep $\psi$ low. 

\begin{figure}[!htb]
\minipage{\textwidth}
\centering
  \includegraphics[width=\linewidth]{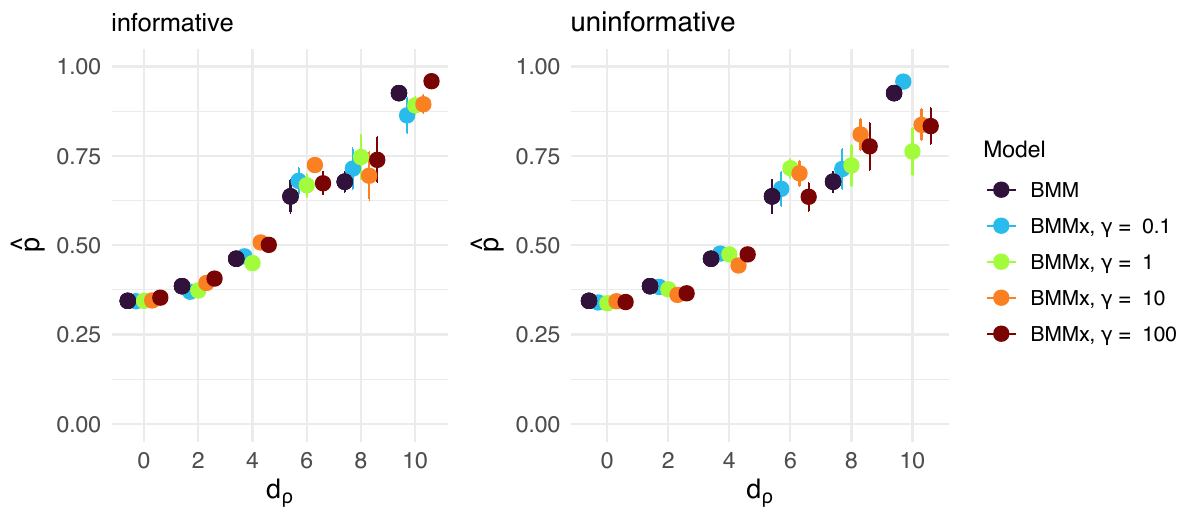}
\endminipage 
\caption{Scatterplot of $\hat{p}$ with variance bars over 10 runs for BMM and BMMx for varying $d_x$ on the x-axis and varying $\gamma$ (color bars), with the goodness-of-fit similarity function. $K=1$ categorical covariate for two levels of informativeness.}
\label{fig:sim_tuning_gamma}
\end{figure}

\begin{figure}[!htb]
\minipage{\textwidth}
\centering
  \includegraphics[width=\linewidth]{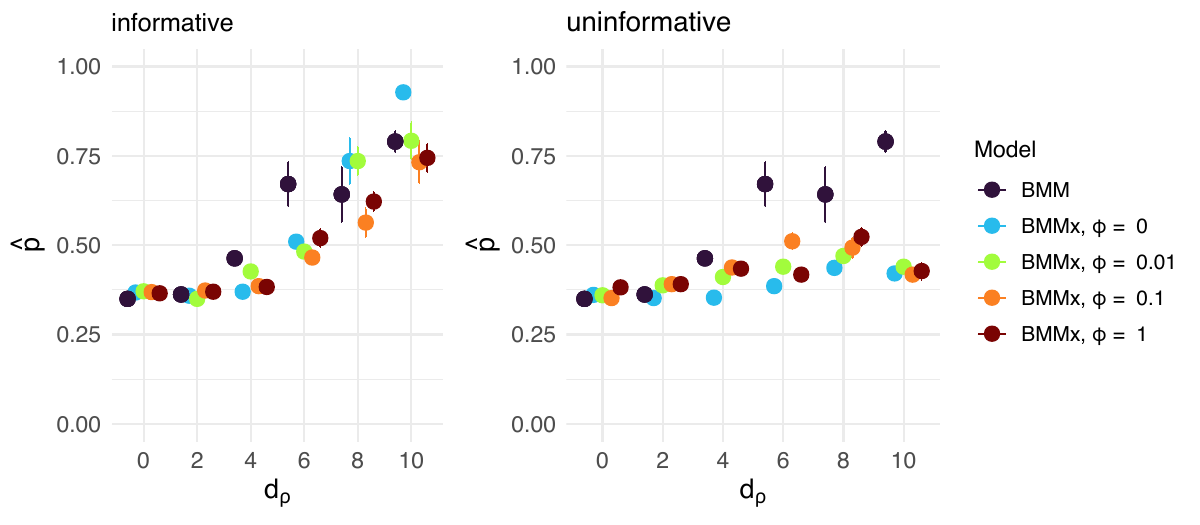}
\endminipage 
\caption{Scatterplot of $\hat{p}$ with variance bars over 10 runs for BMM and BMMx for varying $d_x$ on the x-axis and varying $\phi$ (color bars), with the augmented similarity function. $K=1$ categorical covariates, for two levels of informativeness.}
\label{fig:sim_tuning_psi}
\end{figure}

\section{Case studies}\label{sec:case_study}
This section contains data experiments from two different case studies, demonstrating the versatility of BMMx and its performance across different data structures and applications. Section \ref{sec:case_study_sushi} is focused on a benchmark dataset of sushi preferences from Japan \cite{kamishima2003}, while Section \ref{sec:case_study_breast_cancer} focuses on an RNAseq dataset from breast cancer patients (TCGA Data Portal). 

\subsection{Sushi data}\label{sec:case_study_sushi}

We illustrate clustering based on full rankings using the benchmark data set of sushi preferences collected across Japan \cite{kamishima2003}. $N = 5000$ people were interviewed, each giving a complete ranking of $n = 10$ sushi variants. Additionally, a set of covariates were reported, such as age, gender, and several covariates describing geographical regions the people had lived for parts of their lives (covariates are listed in Table \ref{tab:sushi_covariates}). Cultural differences among Japanese regions influence food preferences, so we expect the assessors to be clustered according to different shared consensus rankings. Additionally, we expect the geographical covariates to possibly be informative covariates for BMMx. 

\begin{table}[!htb]
\centering
\begin{tabular}{l|l|l}
\hline
\textbf{Label} & \textbf{Description} & \textbf{Type} \\ 
   \hline
V2 & gender, 0:male 1:female & categorical \\
V3 & age, 0:15-19  1:20-29   2:30-39   3:40-49   4:50-59    5:60- & categorical \\
V4 & total time to fill questionnaire form & continuous \\
V6 & regional ID most longly lived until 15 years old & categorical \\
V7 & east/west ID (0: east, 1: west) most longly lived until 15 years old & categorical \\
V9 & regional ID at which you currently live & categorical \\
V10 & east/west ID at which you currently live & categorical \\
  \hline
\end{tabular}
\caption{\label{tab:sushi_covariates} Description of the covariates included in the sushi dataset.}
\end{table}

For running BMMx, we select a subset of $N=1000$ assessors out of the total 5000, making sure that the subset is a good representation of the complete dataset (see Figure \ref{supp:sushi_cov_all} and Figure \ref{supp:sushi_cov_subset} in the supplementary material for descriptive plots of the covariates for the full set of assessors and the subset). The subsampling was done due to scalability issues in BMMx when $N$ is large, which is not uncommon for Product Partition models \cite{page_quintana2018, yang2020}. We ran BMM and BMMx (only using the goodness-of-fit similarity function with $\theta=1$ for all continuous covariates and $\gamma=1$ for all categorical covariates) for a range of possible number of clusters $C \in \{1, \ldots, 10\}$, with $M=10^4$ MCMC iterations, and discarded the first 1,000 iterations as burnin. 

For each $C$ we used the MCMC samples to compute the posterior footrule distance between $\bm{\rho}_c$ and the ranking of each assessor assigned to that cluster, and then summed over all assessors and cluster centers: $\sum_{c=1}^C \sum_{j:z_j=c} d(\bm{R}_j, \bm{\rho}_c)$. The posterior distribution of this quantity was then used for choosing the appropriate value for $C$, see Figure \ref{fig:sushi_elbow}. We found an elbow for BMMx at $C=6$, which was then used to further inspect the results.

\begin{figure}[!htb]
\minipage{\textwidth}
\centering
\includegraphics[width=0.8\linewidth, height=6cm, keepaspectratio]{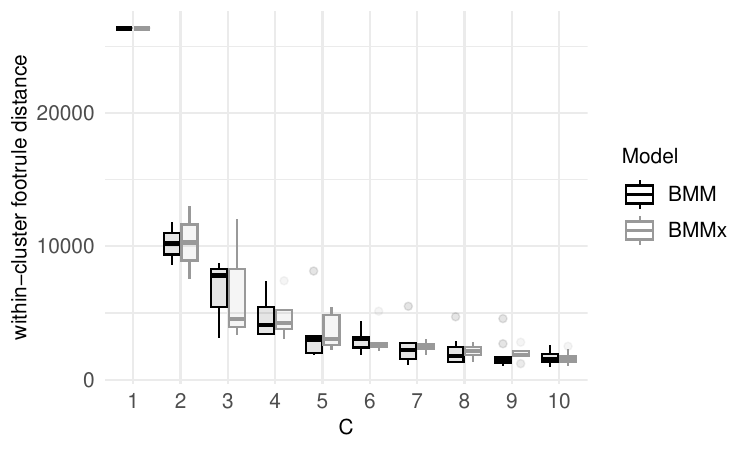}
\endminipage 
\caption{Results of the Sushi experiment. Boxplots of the posterior distributions of the within-cluster sum of footrule distances of assessors' ranks from the corresponding cluster consensus for different choices of $C$ for BMM and BMMx.}
\label{fig:sushi_elbow}
\end{figure}

We investigate the stability of the clustering for BMMx in Figure \ref{fig:sushi_cluster_assignment_pp_bmmx}, which shows the heatplot of the posterior probabilities, for all 1000 assessors (on the x-axis), of being assigned to each of the 6 clusters: most of these individual probabilities were concentrated on some particular preferred value of $c$ among the six possibilities, indicating a reasonably stable behavior in the cluster assignments. For BMM the elbow seemed to be at $C=5$, and a heatplot of the posterior probabilities for all 1000 assessors being assigned to each of the 5 clusters is displayed in Figure \ref{fig:sushi_cluster_assignment_pp_bmm}. If comparing the results from the two methods with $C=6$, the two methods seem to agree on a few of the clusters, as evident from the comparison of cluster assignments reported in Table \ref{tab:sushi_contingency_C6}. 

The MCMC algorithms of both methods showed low acceptance probabilities for $\bm{\rho}$ and $\alpha$ (2,1$\%$ and 4,5$\%$, respectively, for BMMx), but both MCMC chains showed clear convergence, especially for $\alpha$ and $\bm{\tau}$ (see Figures \ref{supp:sushi_alpha_convergence} and 
\ref{supp:sushi_cluster_convergence} in the supplementary material). The low acceptance rate can be due to the ``stickiness'' already characterizing the chains when using the original BMM, where the clustering procedure struggled to accept new proposals when $n << N$. This issue can potentially be mitigated with an adaptive proposal strategy \cite{andrieu2008} for the MCMC or via posterior tempering \cite{marinari1992, geyer_thompson1995}, but this is out of the scope of this paper.

\begin{figure}[!htb]
\minipage{\textwidth}
\centering
  \includegraphics[width=\linewidth]{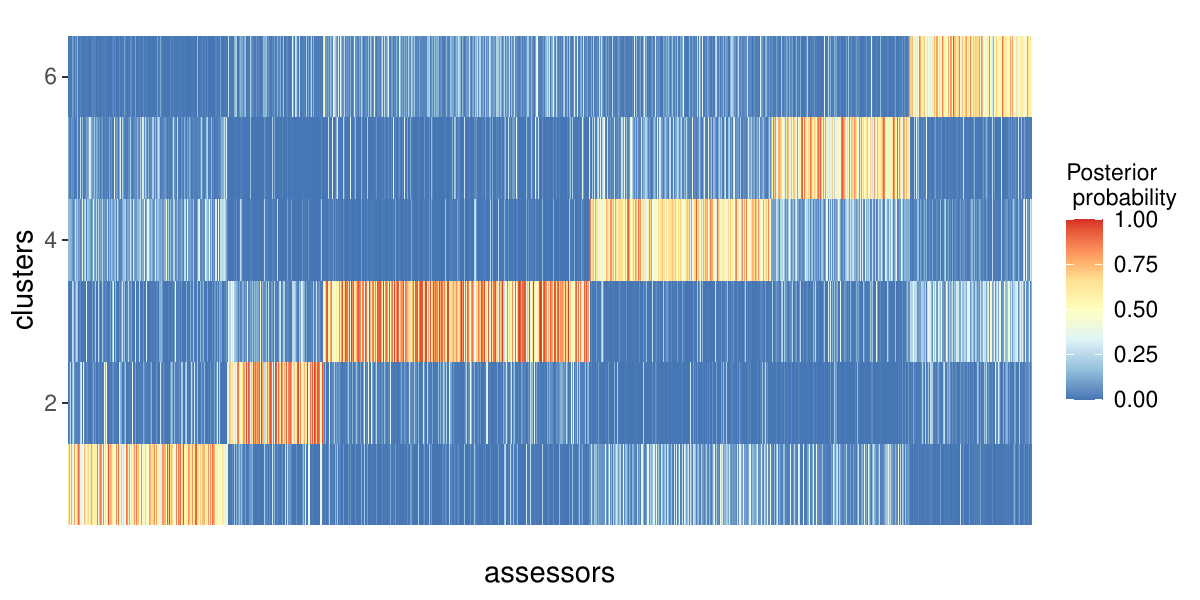}
\endminipage 
\caption{Results of the Sushi experiment. Heatplot of posterior probabilities for all 1000 assessors (on the x-axis) of being assigned to each cluster ($c = 1, \ldots, 6$ from bottom to top), for BMMx.}
\label{fig:sushi_cluster_assignment_pp_bmmx}
\end{figure}

\begin{figure}[!htb]
\minipage{\textwidth}
\centering
  \includegraphics[width=\linewidth]{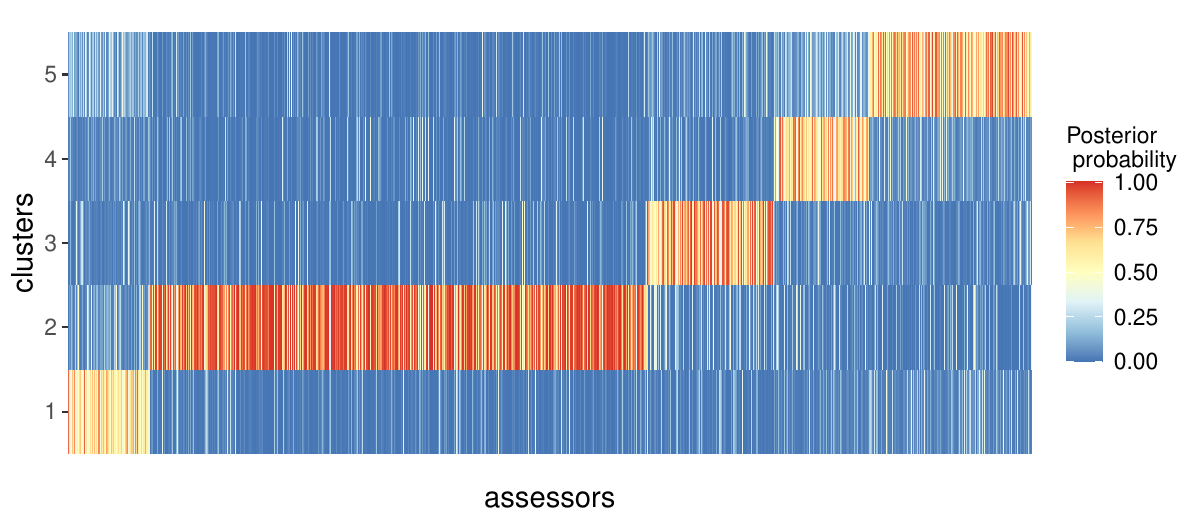}
\endminipage 
\caption{Results of the Sushi experiment. Heatplot of posterior probabilities for all 1000 assessors (on the x-axis) of being assigned to each cluster ($c = 1, \ldots, 5$ from bottom to top), for BMM.}
\label{fig:sushi_cluster_assignment_pp_bmm}
\end{figure}

\begin{table}[ht]
\centering
\begin{tabular}{r|rrrrrr}

  & \multicolumn{6}{c}{BMM}\\
  \hline
BMMx & 1 & 2 & 3 & 4 & 5 & 6 \\ 
  \hline
1 & 108 &  18 &   6 & 174 & 690 &   0 \\ 
  2 & 552 &  12 &   0 &   6 &  18 &   6 \\ 
  3 &  60 & 978 & 444 &  42 &  24 & 114 \\ 
  4 &   0 &   0 &   0 & 240 & 858 &  24 \\ 
  5 &   0 &  12 &   0 & 780 &  54 &  18 \\ 
  6 &  24 &  42 &  78 &  24 &  36 & 558 \\ 
   \hline
\end{tabular}
\caption{Results of the Sushi experiment. Contingency table for BMM and BMMx ($C=6$). }
\label{tab:sushi_contingency_C6}
\end{table}

In order to inspect the distribution of covariates in the different clusters, we plot all $K=7$ covariates according to the estimated cluster assignment (see the supplementary material for all plots: \ref{supp:sushi_cluster_assignment_V2}, \ref{supp:sushi_cluster_assignment_V3}, \ref{supp:sushi_cluster_assignment_V4}, \ref{supp:sushi_cluster_assignment_V7}, \ref{supp:sushi_cluster_assignment_V10}). We here focus on covariates $V6$ and $V9$ only, as they describe the regions of Japan in which the assessors were children or live currently, respectively, and therefore we assume these could be the most informative covariates for the clustering. The proportion of individuals from each cluster assigned by BMMx in each level of covariate $V6$ and $V9$ is displayed in Figure \ref{fig:sushi_V6_V9_unstacked_barplot_bmmx}, with the regions on the x-axis ordered from north to south in Japan. First of all, there is overall not a drastic change in the clustering from V6 to V9, suggesting that the preferences the assessors had as children remain more or less the same when they are adults. Figure \ref{fig:sushi_V6_V9_unstacked_barplot_bmmx}  also shows that clusters 5 and 6 seem to be slightly over-represented in northern-central regions, while cluster 1 is particularly strong in Okinawa. The posterior distribution of the cluster consensus in each of these groups can be inspected to further understand food preferences in relation to regional characteristics.


\begin{figure}[!htb]
\minipage{\textwidth}
\centering
  \includegraphics[width=\linewidth]{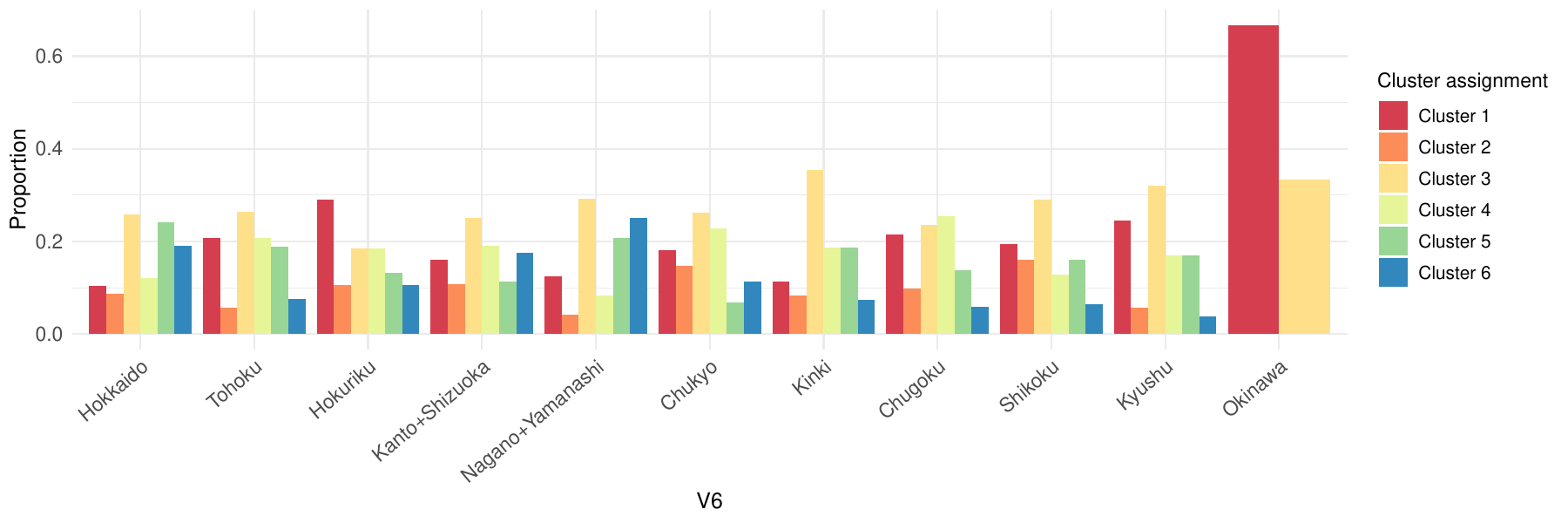}
\endminipage 
\vspace{0.1cm}
\minipage{\textwidth}
\centering
  \includegraphics[width=\linewidth]{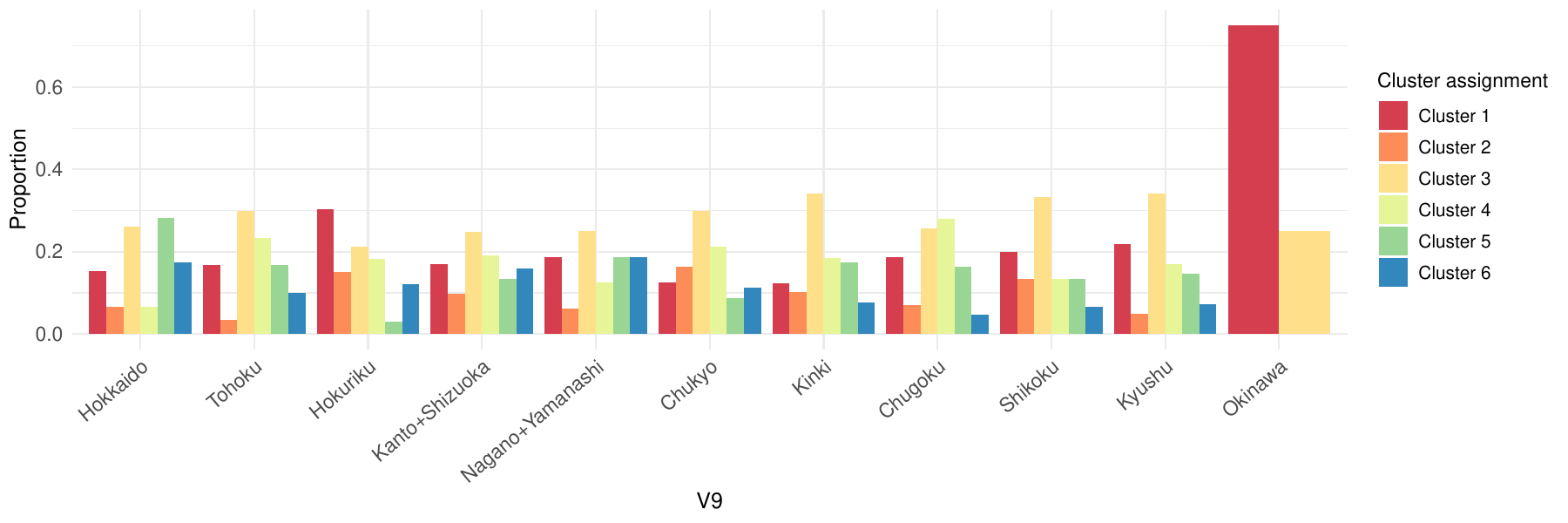}
\endminipage 
\caption{Results of the Sushi experiment. Distribution of the proportion of categories of covariate V6 (top) and covariate V9 (bottom) for the assessors belonging to the different clusters detected by BMMx, when $C=6$. Regions (x-axis) are ordered from north to south in Japan.}
\label{fig:sushi_V6_V9_unstacked_barplot_bmmx}
\end{figure}

\subsection{Breast cancer data}\label{sec:case_study_breast_cancer}
In this section we describe the application of BMMx to RNAseq data from breast cancer patients. Unsupervised learning techniques such as BMMx can provide valuable insights into omics data by facilitating the exploration of subpopulations within the dataset. For instance, in the context of cancer research, unsupervised learning approaches can effectively cluster patients with a particular cancer type into distinct omics-related subtypes. Specifically, the aim of the analysis presented is to identify subgroups of breast cancer patients showing similar gene expression patterns, to gain a better understanding of the molecular drivers of the disease specific to the different subtypes. 

The RNAseq bulk data and clinical data from breast cancer patients were available at the TCGA Data Portal\footnote{https://tcga-data.nci.nih.gov/tcga/}, which provides a large collection of data sets for different cancer types, including both clinical information on the patients, and multi-omics molecular data. The data was downloaded using the R package \texttt{TCGAbiolinks} \cite{mounir2019}. Preprocessing of the RNAseq data included normalization of the RNA count data with DESeq2 \cite{deseq2} followed by selection of the PAM50 genes \cite{sorlie2001, parker2009}, which should represent a highly informative group of genes for the purpose of breast cancer subtyping. We thus obtained a dataset including the RNAseq expression measurements for $n=49$ genes (gene ESR1 was omitted, and instead used as a covariate) and for $N=1047$ patients. In addition, we obtained clinical data where 6 clinical covariates were measured: age, ethnicity, PAM50 subtype (the most accepted partition of breast cancer patients into subtypes), treatments, time, status (see Table \ref{tab:brca_covariates} for further details on the covariates used in the analysis). When running the analysis with BMMx, we omitted the PAM50 subtype from the set of covariates, for using it instead after the analysis in the interpretation of the estimated groups. We also included four additional continuous covariates in our analysis: the RNAseq measurements for the genes ESR1 \cite{holst2007estrogen}, BRCA1, BRCA2 \cite{mavaddat2013cancer} and TP53 \cite{borresen2003tp53, walsh2006spectrum}, which are well-studied genes in breast cancer. We thus ended up with a set of $K=9$ covariates. Before running the analysis, we converted the patients' RNAseq measurements to ranks by ranking each patient data vector from $1$ to $n$ according to the gene expression value (i.e., the highest gene expression value for a given patient gets a 1, and the lowest value gets $n$). The analysis was run with $M=2.5\cdot10^4$ MCMC iterations, with 2,500 iterations discarded as burn-in, which showed to be more than enough for convergence (see Figures \ref{supp:BRCA_alpha_convergence}, \ref{supp:BRCA_rho_convergence} and \ref{supp:BRCA_cluster_convergence} in the supplementary material for convergence plots of $\alpha$, $\bm{\rho}$ and $\bm{\tau}$, respectively). As happened in the sushi case study, acceptance probabilities are low for both BMMx (5$\%$ for $\bm{\rho}$ and 1.6$\%$ for $\alpha$), however slightly higher in this experiment as compared to the sushi, since $n$ is now larger. Similarly, the acceptance probabilities for BMM are 4,5$\%$ for $\bm{\rho}$ and 1.6$\%$ for $\alpha$.

\begin{table}[!htb]
\centering
\begin{tabular}{l|l|l}
\hline
\textbf{Label} & \textbf{Description} & \textbf{Type} \\ 
   \hline
age &  age at diagnosis & continuous \\
ethnicity & "not hispanic or latino", "hispanic or latino", "not reported" & categorical \\
treatments & yes/no & categorical \\
time & observed time since diagnosis & continuous \\
status & dead/alive & categorical \\
ESR1 & RNAseq measurement & continuous \\
BRCA1 & RNAseq measurement & continuous \\
BRCA2 & RNAseq measurement & continuous \\
TP53 & RNAseq measurement & continuous \\
  \hline
\end{tabular}
\caption{\label{tab:brca_covariates} Description of the covariates used in the breast cancer data experiment.}
\end{table}

We ran BMM and BMMx (only with the goodness-of-fit similarity function, with $\theta=1$ for all continuous covariates and $\gamma=0.5$ for all categorical covariates) for a range of possible number of clusters $C \in \{1, \ldots, 10\}$, resulting in the elbow-plot shown in Figure \ref{fig:BRCA_elbow}. The figure shows an elbow at $C=4$, which we examined further for both methods. The two methods yield different cluster estimates, with the highest agreement in clusters 3 and 4,  as evident in the contingency table in Table \ref{tab:BRCA_contingency_C4}.


\begin{figure}[!htb]
\minipage{\textwidth}
\centering
  \includegraphics[width=0.8\linewidth, height=6cm, keepaspectratio]{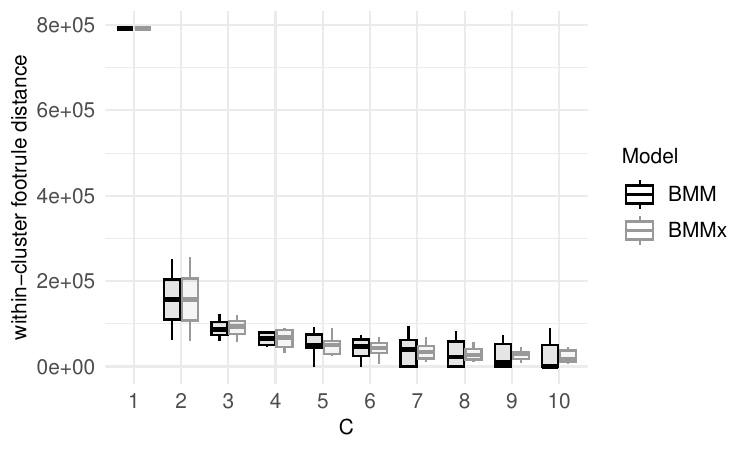}
\endminipage 
\caption{Results of the breast cancer data experiment. Boxplots of the posterior distributions of the within-cluster sum of footrule distances of assessors' ranks from the corresponding cluster consensus for different choices of $C$ for BMM and BMMx.}
\label{fig:BRCA_elbow}
\end{figure}

\begin{table}[!htb]
\centering
\begin{tabular}{crrrr}
& \multicolumn{4}{c}{BMM}\\
  \hline
BMMx & 1 & 2 & 3 & 4 \\ 
  \hline
1 &  1228 & 220 &   0 &   4 \\  
  2 &  0 & 424 &  12 &   0 \\
  3 &  0 &   0 & 764 &   0 \\  
  4 &   40 &   0 &   8 & 1488 \\ 
   \hline
\end{tabular}
\caption{Results of the beast cancer data experiment. Contingency table for BMM and BMMx ($C=4$).}
\label{tab:BRCA_contingency_C4}
\end{table}

The robustness of the clustering results was assessed by analyzing the stability of the cluster assignments. Figure \ref{fig:BRCA_cluster_assignment_pp} shows the heatplot of the posterior probabilities of all 1047 assessors (on the x-axis) being assigned to each of the 4 clusters, with BMMx and BMM. It is observed that a significant portion of these probabilities is concentrated around specific values of $c$ among the four possible clusters. This concentration indicates a consistent and stable behavior in the cluster assignments, reinforcing the reliability of the BMMx algorithm in capturing the underlying patterns in the data.

\begin{figure}[!htb]
\minipage{\textwidth}
\centering
  \includegraphics[width=\linewidth]{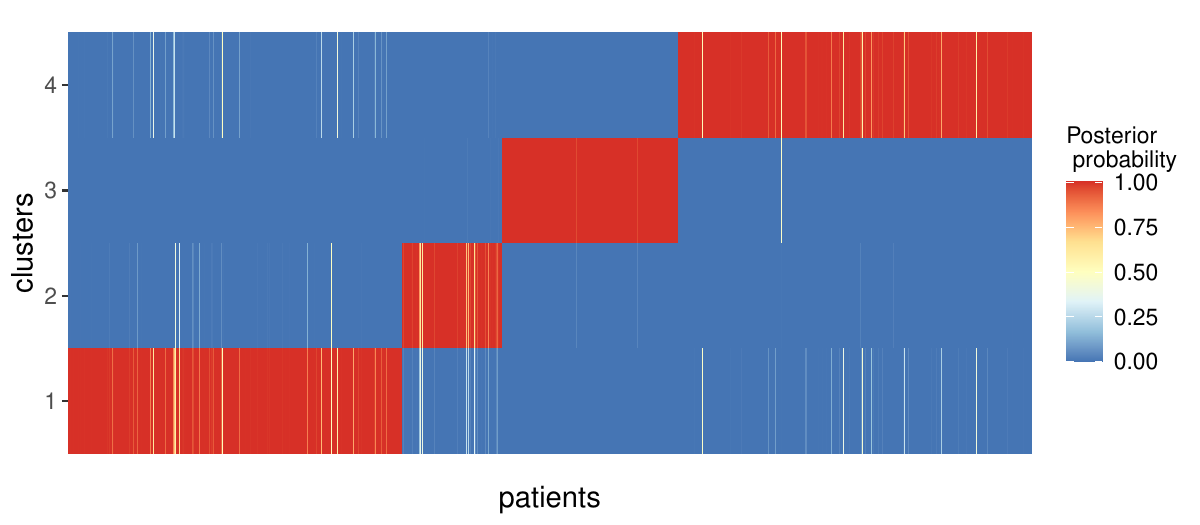}
\endminipage 
\vspace{0.1cm}
\minipage{\textwidth}
\centering
  \includegraphics[width=\linewidth]{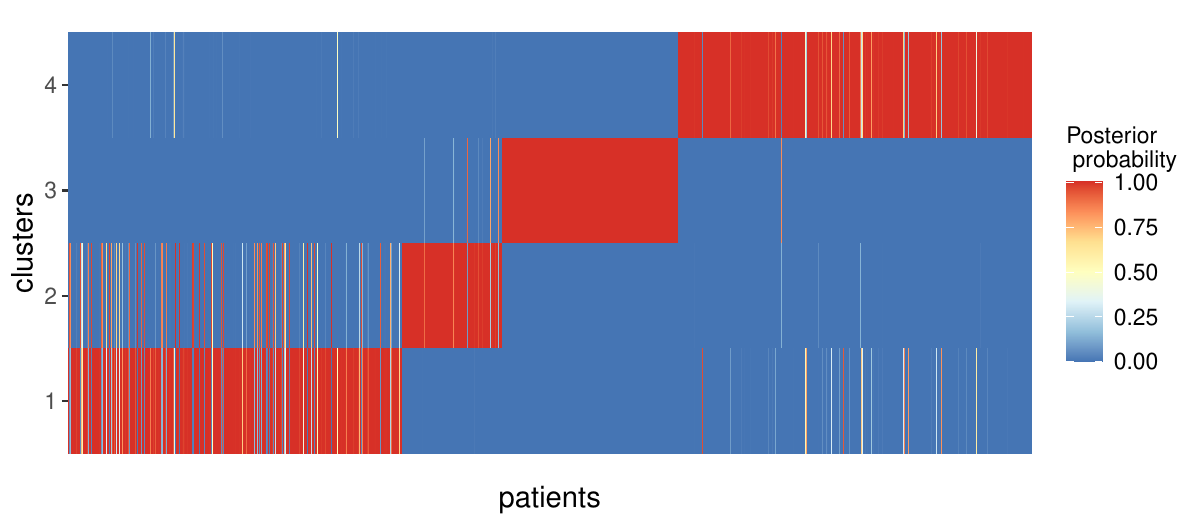}
\endminipage 
\caption{Results of the breast cancer data experiment. Heatplot of posterior probabilities for all $N=1047$ assessors (on the x-axis) of being assigned to each cluster ($c = 1, \ldots, 4$ from bottom to top), for BMMx (top) and BMM (bottom). }
\label{fig:BRCA_cluster_assignment_pp}
\end{figure}

To visualize the covariate distribution across the clusters, we generate plots for all clinical covariates based on the estimated cluster assignments (see the supplementary material for all plots: 
\ref{supp:BRCA_cluster_assignment_age}
\ref{supp:BRCA_cluster_assignment_ethnicity}, \ref{supp:BRCA_cluster_assignment_status}, \ref{supp:BRCA_cluster_assignment_time}, \ref{supp:BRCA_cluster_assignment_treatments}). Even more interestingly, we can plot the PAM50 subtypes according to the cluster assignment, and see if there are any clear differences in the distribution of patient subtypes among the different clusters. There are four widely accepted intrinsic breast cancer subtypes: Luminal A (LumA), Luminal B (LumB), Her2-enriched (Her2) and basal-like (Basal). The originally defined normal-like (Normal) breast cancer subtype is now less frequently used \cite{weigelt2010, elloumi2011, bastien2012}, as it is less characterized and homogeneous, and often associated to the Luminal breast cancer subtypes \cite{rouzier2005breast}. We see in Figure \ref{fig:BRCA_pam50} that patients with the subtype Basal are predominantly assigned to cluster 3, and Her2 to cluster 2. LumA is split among cluster 1 and 4, while LumB mainly appears in cluster 1. Finally, the normal-like cancers are mainly appearing in cluster 4, but also spread between cluster 2 and 3. It is not uncommon to group the two luminal breast cancers together \cite{aure_vitelli2017}. Comparing the results of the clustering in BMMx and BMM, BMMx had a misclassification rate of 0.246 with respect to the PAM50 subtypes, compared to BMM with 0.293 (see the contingency tables in Table \ref{tab:BRCA_confusion_PAM50}). Overall, BMMx performs better in clustering LumB patients, as compared to BMM. BMMx also estimates a better clustering of the normal-like patients. On the other hand, BMM is able to cluster the Her2 subtypes slightly better than BMMx.

\begin{figure}[!htb]
\minipage{0.5\textwidth}
\centering
  \includegraphics[width=\linewidth, height=7cm, keepaspectratio]{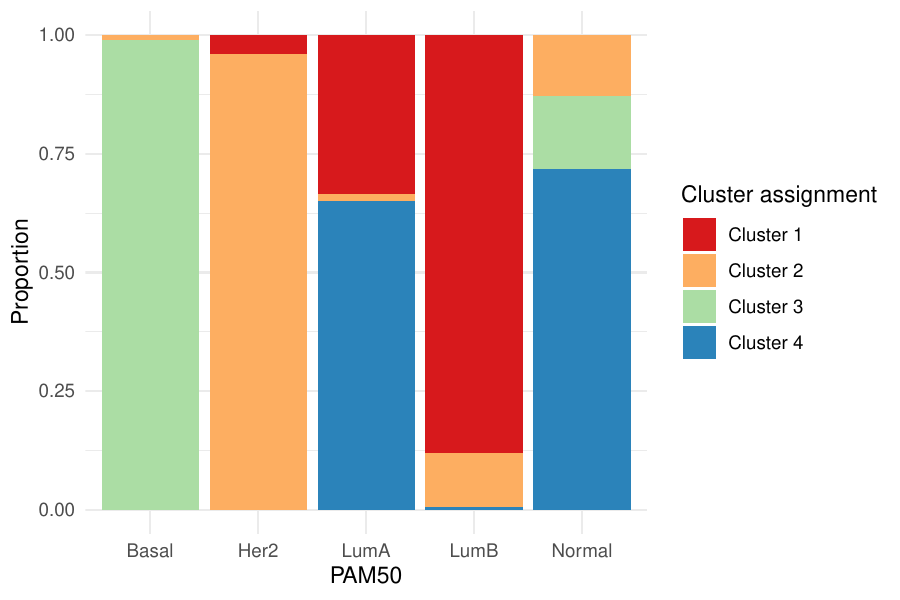}
\endminipage 
\minipage{0.5\textwidth}
\centering
  \includegraphics[width=\linewidth, height=7cm, keepaspectratio]{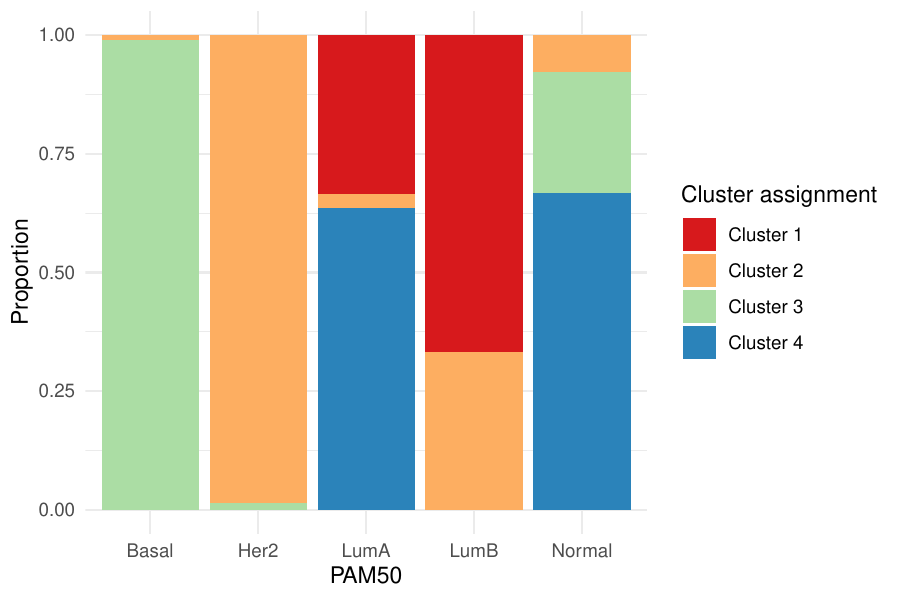}
\endminipage 
\caption{Results of the breast cancer data experiment. Distribution of PAM50 subtypes according to cluster assignment from BMMx (left) and BMM (right) with $C=4$.}
\label{fig:BRCA_pam50}
\end{figure}

\begin{table}[!htb]
\centering
\begin{tabular}{rrrrrr|rrrrr}
  \hline
  & \multicolumn{5}{c|}{BMMx} & \multicolumn{5}{c}{BMM}\\ 
  \hline
 & Basal & Her2 & LumA & LumB & Normal & Basal & Her2 & LumA & LumB & Normal \\  
  \hline
1 &   0 &   3 & 182 & 178 &   0 &    0 &   0 & 182 & 135 &   0 \\  
  2 &   2 &  71 &   8 &  23 &   5 &    2 &  73 &  16 &  67 &   3 \\  
  3 & 185 &   0 &   0 &   0 &   6 &   185 &   1 &   0 &   0 &  10 \\  
  4 &   0 &   0 & 355 &   1 & 28  &   0 &   0 & 347 &   0 &  26 \\ 
   \hline
\end{tabular}
\caption{Results of the breast cancer data experiment. Contingency table for the PAM50 subtypes and BMMx (left) and BMM (right).}
\label{tab:BRCA_confusion_PAM50}
\end{table}

We also plot the distribution of the additional continuous covariates according to cluster assignment from BMMx, with ESR1 and BRCA1 displayed in Figure \ref{fig:BRCA_ESR1_BRCA1} and BRCA2 and TP53 displayed in Figure \ref{fig:BRCA_BRCA2_TP53}. Overall, the ESR1 RNAseq measurements are higher in cluster 3, which is the cluster that contains the most patients with the Basal breast cancer subtype. Cluster 4 (mainly Normal-like subtype) has an overall higher proportion of lower BRCA1 and BRCA2 values. TP53 is more represented in cluster 1, which mainly represents the Luminal B breast cancer subtype.

Additionally, we can compare the two methods by looking at the posterior probability of the genes being ranked top-10 in the consensus parameter $\bm{\rho}_c$ for each cluster, and then by ranking the genes showing such probabilities to be the largest, as displayed in Table \ref{tab:BRCA_top10_bmmx} and Table \ref{tab:BRCA_top10_bmm} for BMMx and BMM respectively. The lists of genes showing the largest top-10 probabilities (in short, from now on, the ``top-10 gene set'') in clusters 1, 3, and 4 for BMMx are quite similar to the top-10 gene sets for BMM. There are however some differences: in cluster 1 the gene BCL-2 appears in the top-10 gene set for BMMx, but not for BMM (for which the gene TYMS appears instead), although the two methods show overall quite similar cluster assignments for this cluster. There is slightly less agreement for cluster 2, for example regarding the top-10 probability for gene BCL-2, which is higher in BMM versus BMMx.  Additionally, the gene CENPF only appears in the top-10 gene set for BMMx, while the gene MELK only appears in the top-10 gene set for BMM (in cluster 2). One can here potentially argue that the top-10 gene set in BMM discriminates the subtypes in cluster 2 slightly better, as evident in Table \ref{tab:BRCA_confusion_PAM50}. On the other hand, BMMx has an overall better estimation of clusters 3 and 4, suggesting that the top-10 gene sets in these clusters characterize them well. Nonetheless, if one looks at the first 8-10 genes, as well as the last 3-5 genes, in the cluster-specific Cumulative Probability (CP) consensus ranking of the genes, the methods conclude quite similarly, as evident in Table \ref{supp:BRCA_CP_consensus} in the supplementary material.

In conclusion, these results show that clustering via BMMx, making use of the available clinical information via a patient-specific prior, can increase the accuracy of the estimated cancer subtypes, as compared to clustering methods (such as BMM) that only make use of molecular information. This highlights the great potential of covariate-informed clustering methods in uncovering hidden patterns in complex biological datasets, by properly combining molecular information with the patient-specific clinical picture.

\begin{figure}[!htb]
\minipage{0.49\textwidth}
\centering
  \includegraphics[width=\linewidth, height=8cm, keepaspectratio]{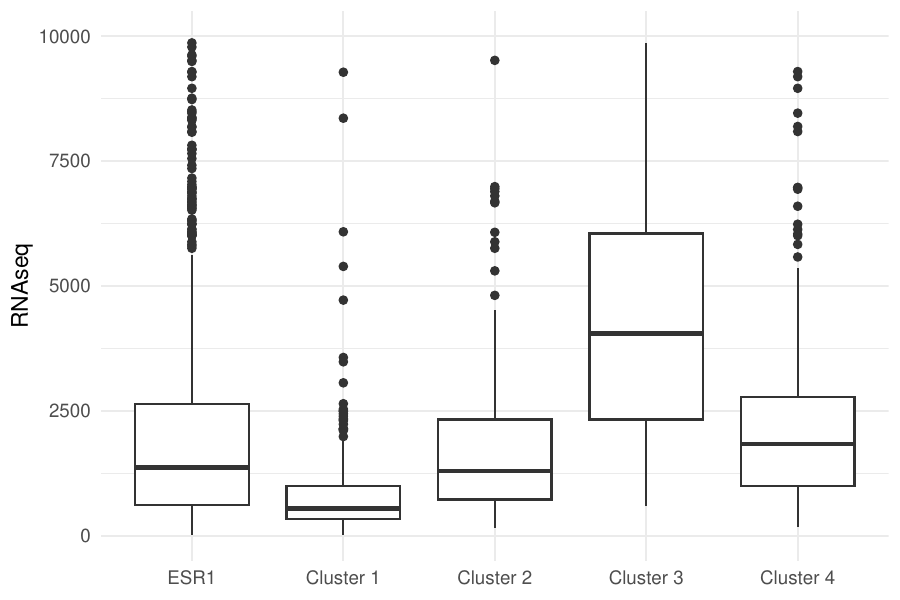}
\endminipage 
\minipage{0.49\textwidth}
\centering
  \includegraphics[width=\linewidth, height=8cm, keepaspectratio]{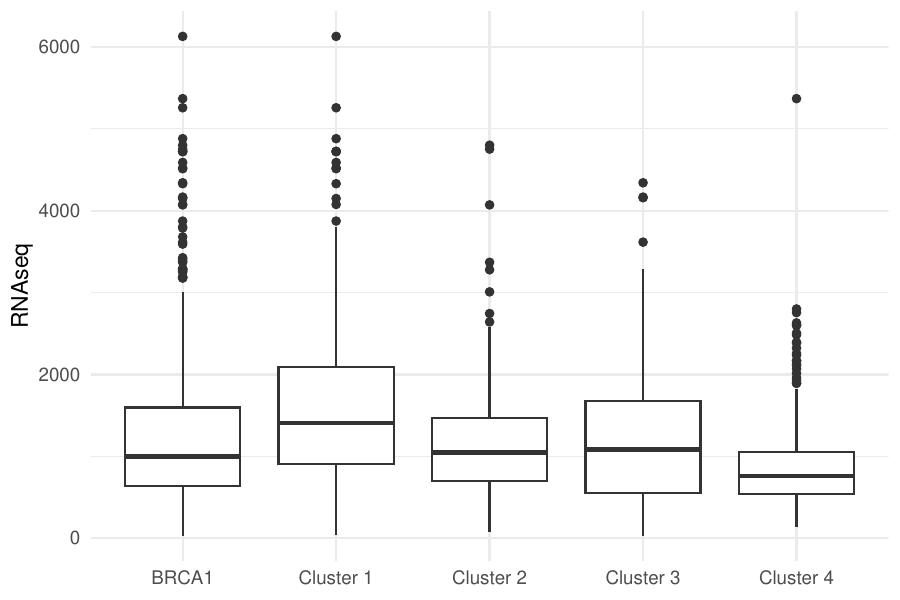}
\endminipage 
\caption{Results of the breast cancer data experiment. Boxplots of ESR1 (left) and BRCA1 (right) according to cluster assignment from BMMx with $C=4$. The left box in each plot represents all measurements (not clustered).}
\label{fig:BRCA_ESR1_BRCA1}
\end{figure}

\begin{figure}[!htb]
\minipage{0.49\textwidth}
\centering
  \includegraphics[width=\linewidth, height=8cm, keepaspectratio]{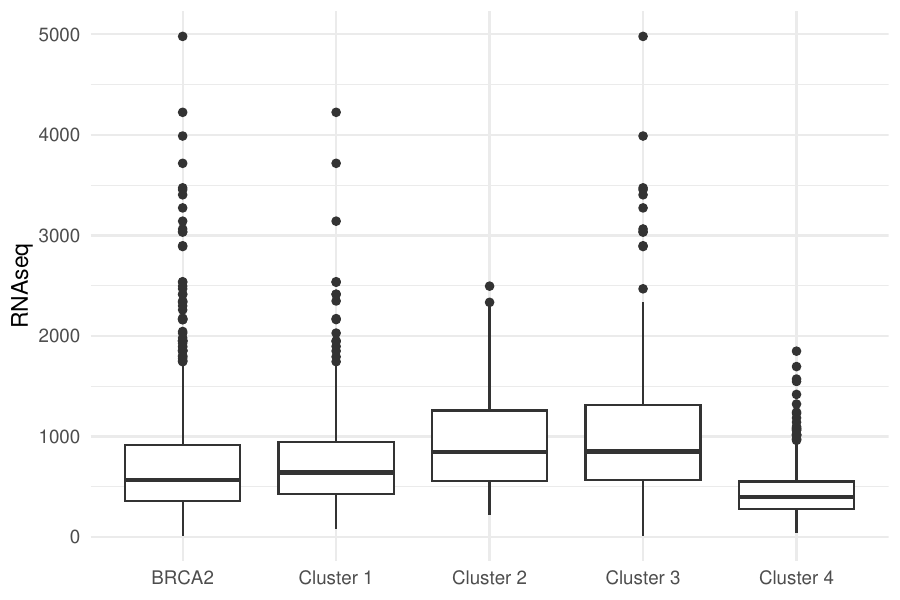}
\endminipage 
\minipage{0.49\textwidth}
\centering
  \includegraphics[width=\linewidth, height=8cm, keepaspectratio]{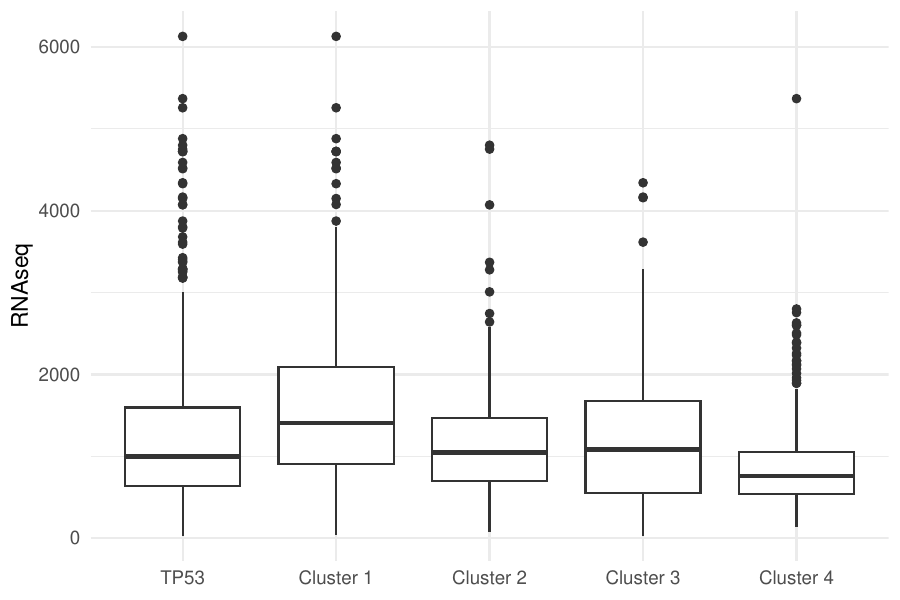}
\endminipage 
\caption{Results of the breast cancer data experiment. Boxplots of BRCA2 (left) and TP53 (right) according to cluster assignment from BMMx with $C=4$. The left box in each plot represents all measurements (not clustered).}
\label{fig:BRCA_BRCA2_TP53}
\end{figure}

\begin{table}[!htb]
\centering
\begin{tabular}{rlr|lr|lr|lr}
  \hline
  & \multicolumn{2}{c|}{Cluster 1} & \multicolumn{2}{c|}{Cluster 2}& \multicolumn{2}{c|}{Cluster 3} & \multicolumn{2}{c}{Cluster 4}\\ \hline
1 & MELK & 1.00 & KRT17 & 1.00 & KRT17 & 1.00 & CCNE1 & 1.00 \\ 
  2 & UBE2C & 1.00 & UBE2C & 1.00 & UBE2C & 1.00 & ACTR3B & 1.00 \\ 
  3 & MIA & 0.99 & MKI67 & 0.99 & MKI67 & 0.99 & TMEM45B & 1.00 \\ 
  4 & ACTR3B & 0.99 & ACTR3B & 0.97 & ACTR3B & 0.97 & UBE2C & 1.00 \\ 
  5 & KRT17 & 0.99 & SLC39A6 & 0.97 & SLC39A6 & 0.97 & MIA & 1.00 \\ 
  6 & TMEM45B & 0.99 & CCNE1 & 0.96 & CCNE1 & 0.96 & MELK & 0.99 \\ 
  7 & CCNE1 & 0.98 & MIA & 0.96 & MIA & 0.96 & KRT17 & 0.99 \\ 
  8 & BCL2 & 0.96 & FOXC1 & 0.94 & FOXC1 & 0.94 & MLPH & 0.97 \\ 
  9 & MLPH & 0.96 & CENPF & 0.79 & CENPF & 0.79 & TYMS & 0.96 \\ 
  10 & MDM2 & 0.87 & BCL2 & 0.44 & BCL2 & 0.44 & FGFR4 & 0.94 \\ 
   \hline
\end{tabular}
\caption{Results of the breast cancer data experiment. Lists of genes showing largest posterior probability of being ranked top-10 in the consensus ranking parameter of each of the clusters, for BMMx.}
\label{tab:BRCA_top10_bmmx}
\end{table}

\begin{table}[!htb]
\centering
\begin{tabular}{rlr|lr|lr|lr}
  \hline
  & \multicolumn{2}{c|}{Cluster 1} & \multicolumn{2}{c|}{Cluster 2}& \multicolumn{2}{c|}{Cluster 3} & \multicolumn{2}{c}{Cluster 4}\\ \hline
1 & CCNE1 & 1.00 & MKI67 & 0.98 & NDC80 & 1.00 & KRT17 & 1.00 \\ 
  2 & ACTR3B & 1.00 & UBE2C & 0.98 & SLC39A6 & 0.99 & MIA & 1.00 \\ 
  3 & BCL2 & 1.00 & ACTR3B & 0.97 & PTTG1 & 0.99 & UBE2C & 1.00 \\ 
  4 & UBE2C & 1.00 & SLC39A6 & 0.96 & MKI67 & 0.99 & MELK & 0.99 \\ 
  5 & MIA & 1.00 & CCNE1 & 0.96 & FGFR4 & 0.99 & MLPH & 0.99 \\ 
  6 & MELK & 0.99 & MIA & 0.95 & NUF2 & 0.99 & TMEM45B & 0.99 \\ 
  7 & TMEM45B & 0.99 & KRT17 & 0.95 & KRT17 & 0.98 & FGFR4 & 0.98 \\ 
  8 & KRT17 & 0.98 & FOXC1 & 0.93 & KIF2C & 0.97 & CCNE1 & 0.98 \\ 
  9 & MLPH & 0.96 & BCL2 & 0.90 & UBE2C & 0.83 & ACTR3B & 0.98 \\ 
  10 & TYMS & 0.86 & MELK & 0.88 & MIA & 0.78 & TYMS & 0.97 \\ 
   \hline
\end{tabular}
\caption{Results of the breast cancer data experiment. Lists of genes showing largest posterior probability of being ranked top-10 in the consensus ranking parameter of each of the clusters for BMM.}
\label{tab:BRCA_top10_bmm}
\end{table}

\section{Discussion and conclusions}\label{sec:discussion}

In this paper, we have developed a novel Mallows-based finite mixture model informed by both the ranking or preference data and assessor-specific covariate information. The method extends the applicability of the Bayesian Mallows ranking model to allow for the inclusion of covariates in the modeling framework. While several rank models have been proposed in the past, the incorporation of covariates has been limited, hindering a comprehensive understanding of the interplay between rankings and assessor characteristics. In contrast, our model bridges this gap, which can potentially lead to richer insights into individual behaviors.

The simulation studies described in Section \ref{sec:simulations} showed that the method performs well on datasets of varying sizes and under varying data-generating procedures, with both continuous and categorical covariates. We applied the proposed model to two different case study datasets, showcasing the method's versatility, which can be applied to any heterogeneous rank dataset with covariate information.

The Mallows model normally works very well with a large number of assessors $N$. However, our approach currently faces limitations in scalability concerning this number, as the current version of BMMx relies on the computation of the similarity of all assessors within each cluster. Additionally, BMM is already limited when the number of items $n$ is very large, which is a problem that has been addressed in the lower-dimensional Bayesian Mallows Model \cite{eliseussen2022} and the PseudoMallows model \cite{liu2021}. The former is a dimension reduction method that performs variable selection within the Bayesian Mallows model, allowing for ultra-high-dimensional datasets, while the latter is an efficient variational approximation to the Bayesian Mallows model, which leads to faster inference while maintaining similar accuracy compared to the exact model. Future research directions could involve the integration of BMMx within these frameworks to enhance scalability, especially in cases involving a large number of items $n$.

Our current implementation ignores label-switching inside the MCMC, as it has previously been shown that doing so is better in achieving full convergence \cite{jasra2005, celeux2000}. For example, the MCMC iterations can be re-ordered after convergence by using Stephen's algorithm \cite{stephens2000}. Nonetheless, a more seamless implementation to handle label-switching in our MCMC methodology warrants future consideration and could be addressed in subsequent enhancements. 

Our MCMC approach does exhibit low acceptance rates, a phenomenon that can be attributed to the ``stickiness problem'' encountered in BMM when the number of items, $n$, is significantly smaller than the number of assessors, $N$. While a low acceptance rate is generally not desired, it does not seem to affect the accuracy of the results in our experiments. One potential solution might involve implementing an adaptive sampling strategy within the MCMC scheme, such as setting a lower $\alpha$ value at the beginning of the chain and increasing it after sufficient exploration. Another possibility is to improve the exploration of the space via tempering, although this falls outside the scope of the current paper.

To our knowledge, there currently exists no Mallows-based method that includes assessor-related covariates.

\section*{Implementation}
The BMMx method is implemented in R/C++ \cite{R, rcpp}, and all scripts are included in the GitHub repository \url{https://github.com/emilieodegaard/BMMx}.

\bibliographystyle{plain}
\bibliography{ms}

\appendix

\section{Derivation of the augmented similarity function $g$ in the case of a continuous covariate}\label{supp:sim_fun_cont}
Here $x_{jk} \sim N(\mu_c, \sigma^2)$ for $j \in S_c, j=1,\ldots,N$ and normal prior $\xi_c = \mu_c \sim N(\mu_0, \sigma_0^2)$.
\begin{align*}
    g(\mathcal{X}_{ck}) & = \int \prod_{j \in S_c} \left[q(x_{jk} | \xi_c) \right] q_{\xi}(\xi_c) \mathrm{d}\xi_c \\
    & = \int \prod_{j \in S_c} \frac{1}{ \sqrt{2 \pi}\sigma} \exp{\left\{- \frac{1}{2\sigma^2} (x_{jk} - \mu_c )^2\right\} }  \frac{1}{\sqrt{2 \pi}\sigma_0} \exp{\left\{- \frac{1}{2\sigma_0^2} (\mu_c -\mu_0)^2\right\} } \mathrm{d}\mu_c \\
    & =  \left( \frac{1}{\sqrt{2\pi \sigma}}\right)^{|S_c|} \frac{1}{\sqrt{2\pi \sigma_0}} \int \exp{\left\{- \frac{1}{2\sigma^2}   \sum_{j \in S_c} \left[ (x_{jk} - \mu_c)^2 \right] - \frac{1}{2\sigma_0^2} (\mu_c - \mu_0)^2\right\} }  \mathrm{d}\mu_c \\  
    & = \left( \frac{1}{\sqrt{2\pi \sigma}}\right)^{|S_c|} \frac{1}{\sqrt{2\pi \sigma_0}} \int \exp{\left\{- \frac{1}{2\sigma^2} \sum_{j \in S_c} \left[ x_{jk}^2 - 2x_{jk} \mu_c + \mu_c^2\right]-\frac{1}{2\sigma_0^2} (\mu_c^2 -2\mu_c\mu_0 + \mu_0^2)\right\} }  \mathrm{d}\mu_c \\ 
    & = \left( \frac{1}{\sqrt{2\pi \sigma}}\right)^{|S_c|} \frac{1}{\sqrt{2\pi \sigma_0}} \exp{\left\{\frac{-\mu_0^2}{2\sigma_0^2} - \sum_{j \in S_c} \frac{x_{jk}^2}{2\sigma^2}\right\} } \cdot \\
    & \hspace{0.5cm} \underbrace{ \int \exp{\left\{  \sum_{j \in S_c} \left[ \frac{-\mu_c^2}{2}\left(\frac{1}{\sigma_0^2} + \frac{1}{\sigma_0^2 |S_c|} \right) + \mu_c \left( \frac{2x_{jk}}{2 \sigma^2} + \frac{2\mu_0}{2\sigma_0^2 |S_c|} \right) \right] \right\} } \mathrm{d}\mu_c }_{\text{$=A$}}
\end{align*}
Let $\Tilde{\mu}_j = \Tilde{\sigma}^2 \left(\frac{x_{jk}}{\sigma^2} + \frac{\mu_0}{\sigma_0^2 |S_c|} \right)$ and $\Tilde{\sigma}^2 = (\frac{1}{\sigma^2}+ \frac{1}{\sigma_0^2 |S_c|})^{-1}$. Focusing on the integral $A$ only:
\begin{align*}
 A & = \int \exp{\left\{  \sum_{j \in S_c} \left[ \frac{-\mu_c^2}{2\Tilde{\sigma^2}} + \frac{\mu_c \Tilde{\mu}^2_j}{\Tilde{\sigma}^2} - \frac{\Tilde{\mu}_j^2}{2\Tilde{\sigma}^2} + \frac{\Tilde{\mu}_j^2}{2\Tilde{\sigma}^2}\right] \right\} } \mathrm{d}\mu_c \\
 & = \exp{\left\{  \sum_{j \in S_c}  \frac{\Tilde{\mu}_c^2}{2\Tilde{\sigma^2}} \right\}}  \int \exp{\left\{\sum_{j \in S_c} -\frac{(\mu_c - \Tilde{\mu}_j)^2}{2\Tilde{\sigma^2}} \right\}} \mathrm{d}\mu_c\\
 & = \exp{\left\{  \sum_{j \in S_c}  \frac{\Tilde{\mu}_c^2}{2\Tilde{\sigma^2}} \right\}} \left(\sqrt{2\pi} \Tilde{\sigma}\right)^{|S_c|}
\end{align*}
Then:
\begin{align*}
    g(\mathcal{X}_{ck}) & = \left( \frac{\Tilde{\sigma}}{\sigma}\right)^{|S_c|} \frac{1}{\sqrt{2\pi \sigma_0}} \exp{\left\{ -\frac{\mu_0^2}{2\sigma_0^2} - \sum_{j \in S_c} \frac{x_{jk}^2}{2\sigma^2} +  \sum_{j \in S_c} \frac{\Tilde{\mu}_j^2}{2\Tilde{\sigma}^2} \right\}} \\
    & \propto  \left( \frac{\Tilde{\sigma}}{\sigma}\right)^{|S_c|} \exp{\left\{ - \frac{1}{2} \sum_{j \in S_c} \left[\frac{x_{jk}^2}{\sigma^2} - \frac{\Tilde{\mu}_j^2}{\Tilde{\sigma}^2} \right]  \right\}}.
\end{align*}
\section{Derivation of the augmented similarity function $g$ in the case of a categorical covariate}\label{supp:sim_fun_categorical}
Here $x_{jk} \sim \mathcal{C}(\nu_1, \ldots, \nu_B)$  for $j \in S_c, j=1,\ldots,N$ with $B$ categories and Dirichlet prior $\xi_c = \nu_c \sim \mathcal{D}(\varphi_1,\ldots,\varphi_B)$. 
\begin{align*}
    g(\mathcal{X}_{ck}) & = \int \prod_{j \in S_c} \left[q(x_{jk} | \xi_c) \right] q_{\xi}(\xi_c) \mathrm{d}\xi_c \\
    & = \int \prod_{j \in S_c} \prod_{b=1}^B \nu_{c,b}^{1_{x_{jk}=b}} \frac{\Gamma(\sum_{b=1}^B \varphi_b)}{\prod_{b=1}^B \Gamma(\varphi_b)}\prod_{b=1}^B \nu_{c,b}^{\varphi_b-1}\mathrm{d}\nu_c \\
    & = \frac{\Gamma(\sum_{b=1}^B \varphi_b)}{\prod_{b=1}^B \Gamma(\varphi_b)} \int \prod_{b=1}^B \nu_{c,b}^{\sum_{j \in S_c} 1_{x_{jk}=b} + \varphi_b -1}\mathrm{d}\nu_c \\
    & = \frac{\Gamma(\sum_{b=1}^B \varphi_b)}{\prod_{b=1}^B \Gamma(\varphi_b)} \prod_{b=1}^B \frac{\Gamma(\sum_{j \in S_c} 1_{x_{jk}=b} + \varphi_b)}{\Gamma( \sum_{b=1}^B \sum_{j \in S_c} 1_{x_{jk}=b} + \varphi_b)} \\
    & \propto \, \prod_{b=1}^B \frac{\Gamma(\sum_{j \in S_c} 1_{x_{jk}=b} + \varphi_b)}{\Gamma( \sum_{b=1}^B \sum_{j \in S_c} 1_{x_{jk}=b} + \varphi_b)}.
\end{align*}

\section{Details on the Gibbs sampling for $z$ in the presence of covariates}\label{supp:gibbs_sampling}

Let $z_{-j}$ be all labels $z$ excluding $z_j$. The full conditional distribution for $z_j$, $j=1,\ldots,N$ can be derived in the following way:
\begin{align*}
    P(z_j | z_{-j}, \bm{\rho}, \alpha, \tau, \bm{R}, \bm{x}) & = \frac{P(z, \bm{\rho}, \alpha, \tau | \bm{R}, \bm{x})}{P(z_{-j}, \bm{\rho}, \alpha, \tau | \bm{R}, \bm{x})} \\
    & = \frac{P(\bm{R} | z, \bm{\rho}, \alpha, \tau, \bm{x}) P(z, \bm{\rho}, \alpha, \tau | \bm{x})}{P(\bm{R} | \bm{x}) P(z_{-j}, \bm{\rho}, \alpha, \tau | \bm{R}, \bm{x})} \\ 
    & = \frac{P(\bm{R} | z, \bm{\rho}, \alpha, \tau, \bm{x}) P(z, \bm{\rho}, \alpha, \tau | \bm{x})}{P(\bm{R} | \bm{x}) \frac{P(\bm{R} | z_{-j}, \bm{\rho}, \alpha, \tau, \bm{x}) P(z_{-j}, \bm{\rho}, \alpha, \tau, \bm{x})}{P(\bm{R} | \bm{x})}} \\ 
    & = \frac{P(\bm{R} | z, \bm{\rho}, \alpha, \tau, \bm{x}) P(z, \bm{\rho}, \alpha, \tau | \bm{x})}{P(\bm{R} | z_{-j}, \bm{\rho}, \alpha, \tau, \bm{x}) P(z_{-j}, \bm{\rho}, \alpha, \tau, \bm{x})} \\ 
    & = \frac{P(\bm{R}_j | z_j, \bm{\rho}, \alpha, \tau, \bm{x}) P(z | \bm{\rho}, \alpha, \tau, \bm{x}) P(\bm{\rho}, \alpha, \tau | \bm{x})}{P(\bm{R}_j | z_{-j}, \bm{\rho}, \alpha, \tau, \bm{x}) P(z_{-j} | \bm{\rho}, \alpha, \tau, \bm{x}) P(\bm{\rho}, \alpha, \tau | \bm{x})} \\ 
    & = \frac{P(\bm{R}_j | z_j, \bm{\rho}, \alpha, \tau, \bm{x}) P(z | \bm{\rho}, \alpha, \tau, \bm{x})}{P(\bm{R}_j | z_{-j}, \bm{\rho}, \alpha, \tau, \bm{x}) P(z_{-j} | \bm{\rho}, \alpha, \tau, \bm{x}) } .
\end{align*}
The first term in the denominator can be rewritten, by applying the law of total probability:
\begin{align*}
    P(\bm{R}_j | z_{-j}, \bm{\rho}, \alpha, \tau, \bm{x}) & = \sum_{z_j \in \{1,\ldots,C\}} \left[ P(\bm{R}_j | z, \bm{\rho}, \alpha, \tau, \bm{x}) P(z_j | z_{-j}, \bm{\rho}, \alpha, \tau, \bm{x} ) \right] \\
    & = \sum_{z_j \in \{1,\ldots,C\}} \left[ P(\bm{R}_j | z, \bm{\rho}, \alpha, \tau, \bm{x}) \frac{P(z | \bm{\rho}, \alpha, \tau, \bm{x} )  }{P(z_{-j} | \bm{\rho}, \alpha, \tau, \bm{x})  } \right] \\
    & = \frac{1}{P(z_{-j} | \bm{\rho}, \alpha, \tau, \bm{x})} \sum_{z_j \in \{1,\ldots,C\}} \left[ P(\bm{R}_j | z, \bm{\rho}, \alpha, \tau, \bm{x}) P(z | \bm{\rho}, \alpha, \tau, \bm{x} )   \right].
\end{align*}
Applying the conditional independence assumptions of the various parameters stated in Section \ref{sec:model_bmmx}, we are left with the following expression for the full conditional distribution for $z_j$:
\begin{equation}\label{eq:z_full_cond}
    P(z_j | z_{-j}, \bm{\rho}, \alpha, \tau, \bm{R}_j, \bm{x}) = \frac{P(\bm{R}_j | z_j, \bm{\rho}, \alpha) P(z | \tau, \bm{x})}{\sum_{z_j \in \{1,\ldots,C\}} \left[ P(\bm{R}_j | z, \bm{\rho}, \alpha) P(z |  \tau, \bm{x} ) \right]}
\end{equation}
where the term $P(z | \tau, \bm{x})$ is as in \eqref{eq:bmmx_cluster_lab_prior} and $P(\bm{R}_j | z_j, \bm{\rho}, \alpha) = Z_n(\alpha_{z_j})^{-1} \exp\{-(\alpha_{z_j}/n) d(\bm{R}_j, \bm{\rho}_{z_j})\}$. Moreover, we have
\begin{align*}
    P(z_j = c | z_{-j}, \bm{\rho}_c, \alpha_c, \tau_c, \bm{R}_j, \bm{x}^*_c) & \, \propto \, P(\bm{R}_j | z_j, \bm{\rho}_c, \alpha_c) P(z | \tau_c, \bm{x}^*_c) \\
    & \, \propto \, Z_n(\alpha_{z_j})^{-1} \exp\{-(\alpha_{c}/n) d(\bm{R}_j, \bm{\rho}_{c})\} \tau_c g(x^*_c)
\end{align*}
where $x^*_{c} = (x_{c}, l : z_l = c \cup \{j\})$ for $l=1,\ldots,N$ (i.e. $x^*_{c}$ is the covariate of the assessors belonging to cluster $c$ \emph{and} the $j$-th assessor). 

\newpage
\section{MCMC algorithm for BMMx in the case of missing data}\label{supp:bmmx_mcmc_missing}
The pseudo-code of the MCMC algorithm in the case of missing data as described in Section \ref{sec:model_mcmc_missing}, is reported here. 
\begin{algorithm}[!htb]
\SetAlgoLined
\scriptsize
\textbf{input:} $\{\mathcal{S}_1,\ldots,\mathcal{S}_N\}$, 
$\mathbf{x}_1,\ldots,\mathbf{x}_N$; $C$, $\psi$, $c_1$, $c_2$, $\varphi$, $\theta$, $\gamma$, $\sigma_\alpha^2$, $\alpha_{\textnormal{jump}}$, $d(\cdot, \cdot)$, $l$, $Z_n(\alpha)$, $M$, $aug$ \\
\textbf{output:} posterior distributions of $\bm{\rho}_1,\ldots,\bm{\rho}_C$, $\alpha_1,\ldots,\alpha_C$, $\tau_1,\ldots,\tau_C$, $z_1,\ldots,z_N$\\
 \textbf{initialization:} randomly generate $\bm{\rho}_{1,0},\ldots,\bm{\rho}_{C,0}$, $\alpha_{1,0},\ldots,\alpha_{C,0}$, $\tau_{1,0},\ldots,\tau_{C,0}$ and $z_{1,0},\ldots,z_{N,0}$ \leavevmode \\ 
\lForEach{$j \gets 1$ \KwTo $N$}{randomly generate $\bm{\tilde{R}}_j^0$ compatible with $\mathcal{S}_j$}
 \For{$m \gets 1$ \KwTo $M$}{
    \textbf{Gibbs step}: update $\tau_1,\ldots,\tau_C$ \\
    compute: $n_c=\sum_{j=1}^N 1_c(z_{j,m-1})$ for $c=1,\ldots,C$
    sample: $\tau_1,\ldots,\tau_C \sim \mathcal{D}(\psi + n_1,\ldots,\psi+n_C)$ \\
    \For{$c \gets 1$ \KwTo $C$}{
        \textbf{M-H step}: update $\bm{\rho}_c$ \\
        sample: $\bm{\rho}_c' \sim \textnormal{LS}(\bm{\rho}_{c,m-1}, l)$ and $u \sim \mathcal{U}(0,1)$ \\
        compute: \textit{ratio} $\gets$ equation \eqref{eq:accept_prob_rho} with $\bm{\rho}_c \gets \bm{\rho}_{c,m-1}$, $\alpha_c \gets \alpha_{c,m-1}$ and $z_j \gets z_{j,m-1}$\\
        \eIf{$u <$ ratio}{$\bm{\rho}_{c,m} \gets \bm{\rho}'_{c}$}{$\bm{\rho}_{c,m} \gets \bm{\rho}_{c,m-1}$}
        \textbf{M-H step}: update $\alpha_c$ \leavevmode \\
        sample: $\alpha_c' \sim \log \mathcal{N}(\log(\alpha_{c,m-1}), \sigma_\alpha^2)$ and $u \sim \mathcal{U}(0,1)$ \leavevmode\\
        compute: \textit{ratio} $\gets $ equation \eqref{eq:accept_prob_alpha} with $\bm{\rho}_c \gets \bm{\rho}_{c,m}$, $\alpha_c \gets \alpha_{c,m-1}$ and  $z_j \gets z_{j,m-1}$\\
        \eIf{$u <$ ratio}{$\alpha_{c,m} \gets \alpha'_{c}$}{$\alpha_{c,m} \gets \alpha_{c,m-1}$}
    }
    \textbf{Gibbs step}: update $z_1,\ldots,z_N$ \leavevmode \\
    \For{$j \gets 1$ \KwTo $N$}{
    \For{$c \gets 1$ \KwTo $C$}{
        \eIf{aug}{
            \lIf{$x$ continuous}{$g_k(\mathcal{X}_{ck})$ $\gets$ equation \eqref{eq:sim_fun_aug_continuous}}
            \lIf{$x$ categorical}{$g_k(\mathcal{X}_{ck})$ $\gets$ equation \eqref{eq:sim_fun_aug_categorical}}
            }
        		{
            \lIf{$x$ continuous}{$g_k(\mathcal{X}_{ck})$ $\gets$ equation \eqref{eq:sim_fun_alternative_cont}}
            \lIf{$x$ categorical}{$g_k(\mathcal{X}_{ck})$ $\gets$ equation \eqref{eq:sim_fun_alternative_cat}}
            }
    compute: $p_{cj} \gets$ equation \eqref{eq:gibbs_clusterlabels} with $\tau_c \gets \tau_{c,m}$, $\alpha_c \gets \alpha_{c,m}$ and $\bm{\rho}_c \gets \bm{\rho}_{c,m}$
    }
    sample: $z_{j,m} \sim \mathcal{M}(p_{1j},\ldots,p_{Cj})$
    } 
    \textbf{M-H step:}: update $\bm{\tilde{R}}_1,\ldots,\bm{\tilde{R}}_N$: \leavevmode \\
    \For{$j \gets 1$ \KwTo $N$}{
    sample: $\bm{\tilde{R}}'_j$ in $\mathcal{S}_j$ from the leap-and-shift distribution centered at $\bm{\tilde{R}}^{m-1}_j$ \\
    compute: \textit{ratio} $\gets $ equation \eqref{eq:accept_data_aumentation} with $\bm{\rho}_{z_j} \gets \bm{\rho}_{z_{j,m},m}$, $\alpha_{z_j} \gets \alpha_{z_{j,m},m-1}$ and  $\bm{\tilde{R}}_j \gets \bm{\tilde{R}}_j^{m-1}$\leavevmode \\
    sample $u \sim \mathcal{U}(0,1)$ \leavevmode \\
    \eIf{$u <$ ratio}{$\bm{\tilde{R}}_j^m \gets \bm{\tilde{R}}_j'$}{$\bm{\tilde{R}}_j^m \gets \bm{\tilde{R}}_j^{m-1}$}
	}
}
\caption{MCMC scheme for inference in BMMx with missing data}\label{alg:bmmx_mcmc_missing}
\end{algorithm}

\newpage


\title{Supplementary material to the paper: 
\\
Rank-based Bayesian clustering via covariate-informed Mallows mixtures}
\author{Emilie Eliseussen, Arnoldo Frigessi and Valeria Vitelli}
\date{\empty}

\maketitle

\renewcommand{\thefigure}{S\arabic{figure}}
\setcounter{figure}{0}
\renewcommand{\thetable}{S\arabic{table}}
\setcounter{table}{0}
\renewcommand{\thesection}{S\arabic{section}}
\setcounter{section}{0}

\section{Sushi case study}

\begin{table}[H]
\centering
\begin{tabular}{|l|l|}
\hline
 \textbf{Category} & \textbf{Region} \\ 
   \hline
 0 & Hokkaido      \\   
 1 & Tohoku        \\   
 2 & Hokuriku      \\    
 3 & Kanto+Shizuoka  \\  
 4 & Nagano+Yamanashi  \\
 5 & Chukyo           \\ 
 6 & Kinki            \\
 7 & Chugoku         \\ 
 8 & Shikoku          \\
 9 & Kyushu          \\  
10 & Okinawa         \\ 
11 & Foreign          \\  
  \hline
\end{tabular}
\caption{Labelling of regions in Japan, from sushi dataset. }\label{supp:sushi_regions}
\end{table}

\begin{figure}[H]
\minipage{\textwidth}
\includegraphics[width=\linewidth]{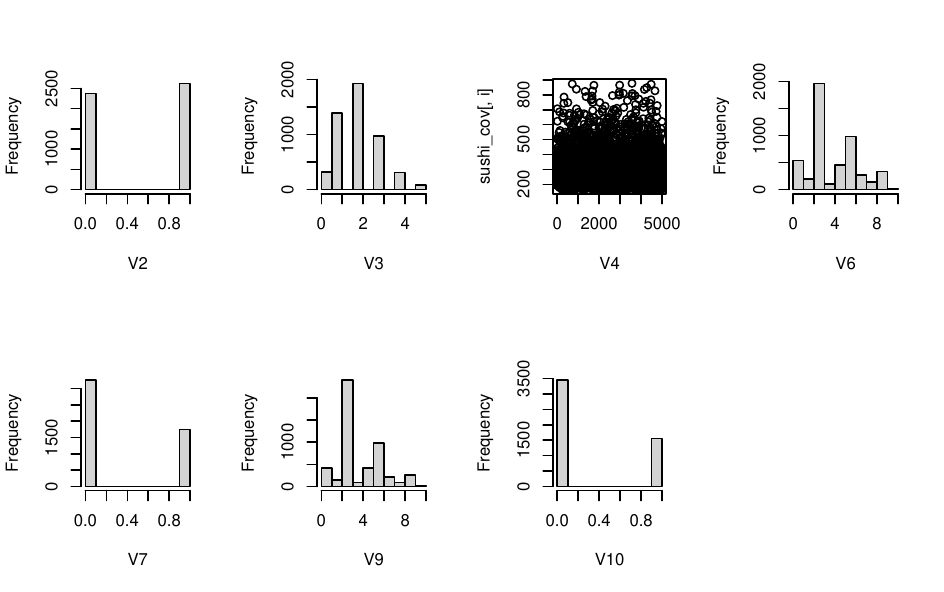}
\endminipage 
\caption{Descriptive plots of the covariates in the full sushi dataset ($N=5000$).}\label{supp:sushi_cov_all}
\end{figure}

\begin{figure}[H]
\minipage{\textwidth}
  \includegraphics[width=\linewidth]{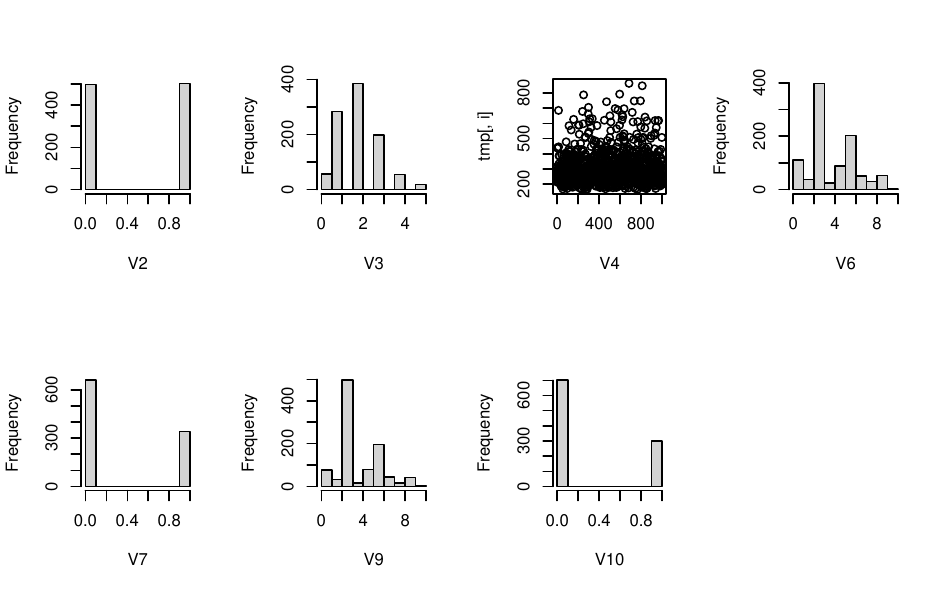}
\endminipage 
\caption{Descriptive plots of the covariates in the subset of the sushi dataset ($N=1000$).}\label{supp:sushi_cov_subset}
\end{figure}

\begin{figure}[H]
\minipage{\textwidth}
\centering
  \includegraphics[width=0.8\linewidth]{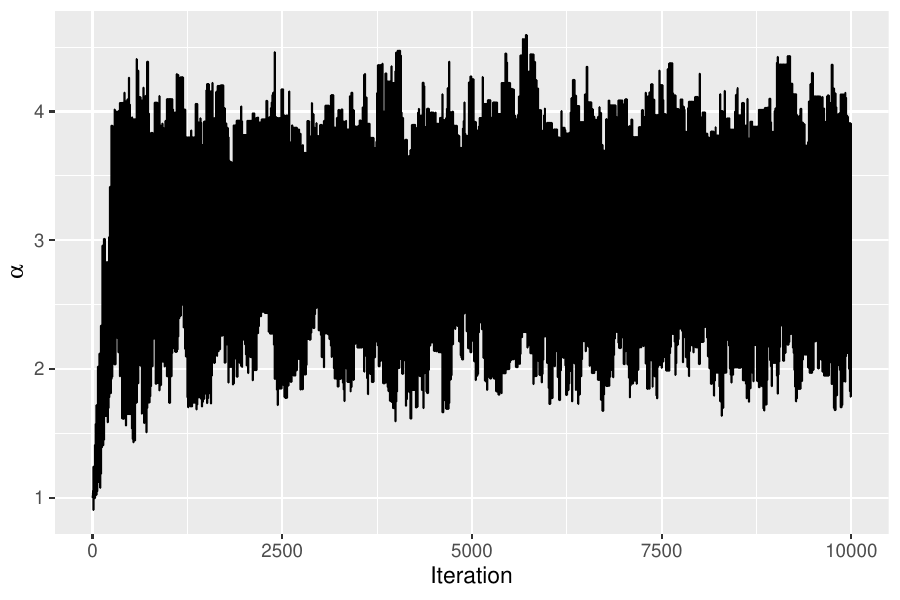}
\endminipage 
\caption{Trace plot of $\alpha$ for the sushi dataset, with $C=6$ from run with BMMx.}
\label{supp:sushi_alpha_convergence}
\end{figure}

\begin{figure}[H]
\minipage{\textwidth}
\centering
  \includegraphics[width=0.8\linewidth]{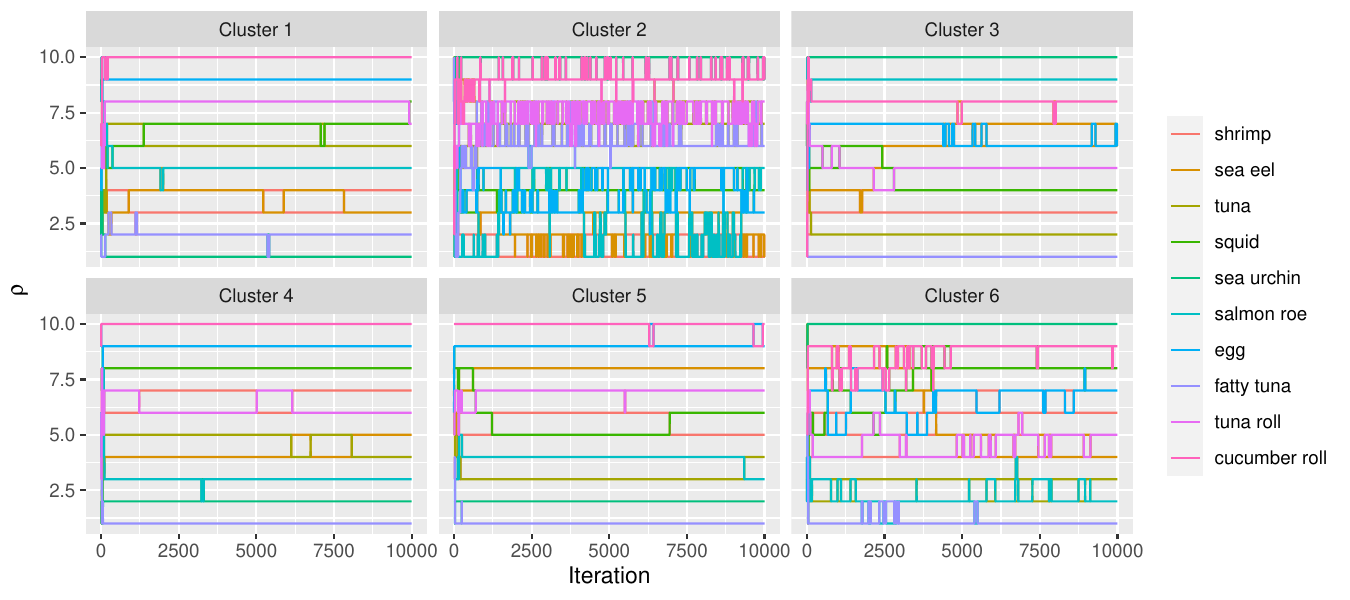}
\endminipage 
\caption{Trace plot of $\mathbf{\rho}_c$ for the sushi dataset, with $C=6$ from run with BMMx.}\label{supp:sushi_rho_convergence}
\end{figure}

\begin{figure}[H]
\minipage{\textwidth}
\centering
  \includegraphics[width=0.8\linewidth]{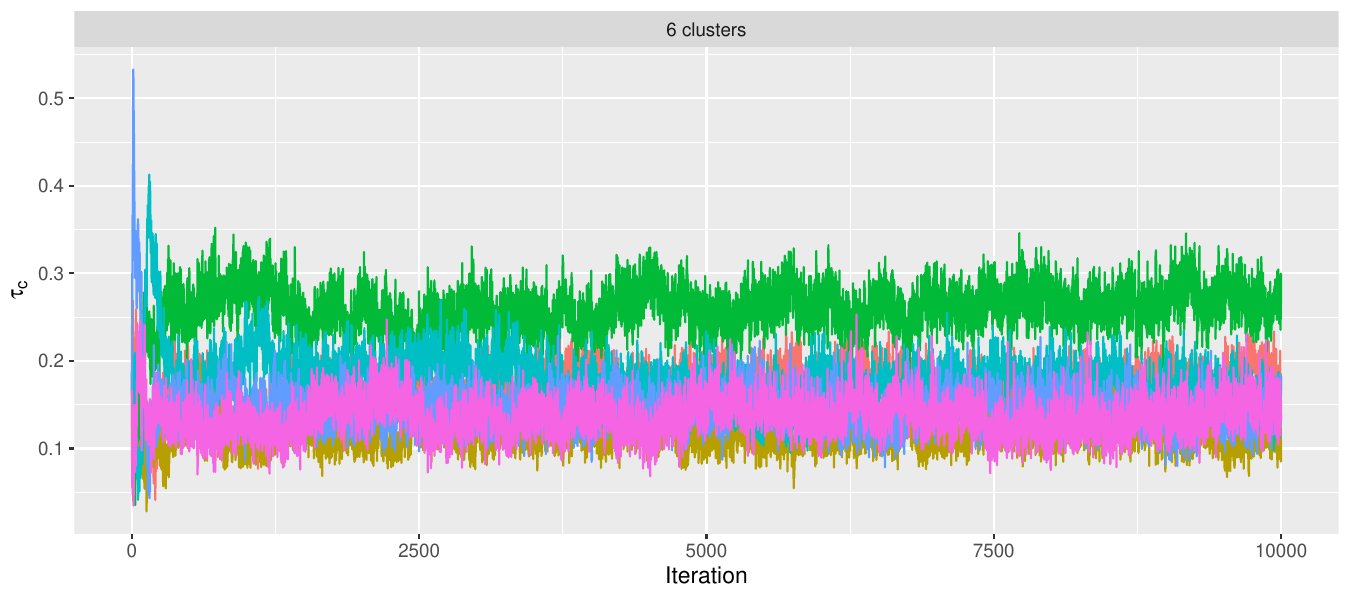}
\endminipage 
\caption{Trace plot of $\mathbf{\tau}_c$ for the sushi dataset, with $C=6$ from run with BMMx.}\label{supp:sushi_cluster_convergence}
\end{figure}

\begin{figure}[H]
\minipage{\textwidth}
\centering
  \includegraphics[width=0.8\linewidth]{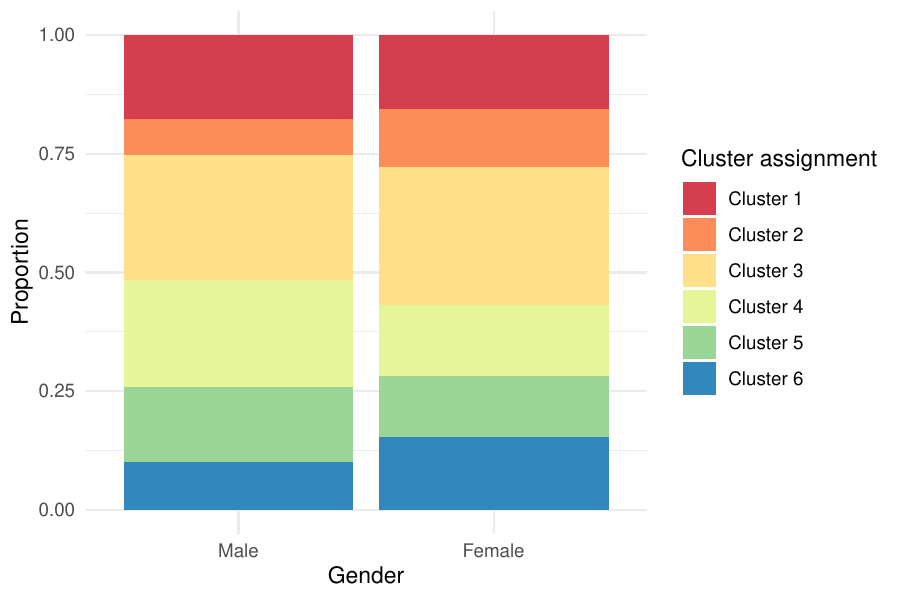}
\endminipage 
\caption{Distribution of covariate V2 according to cluster assignment from run with BMMx with $C=6$ for the sushi dataset. }\label{supp:sushi_cluster_assignment_V2}
\end{figure}

\begin{figure}[H]
\minipage{\textwidth}
\centering
  \includegraphics[width=0.8\linewidth]{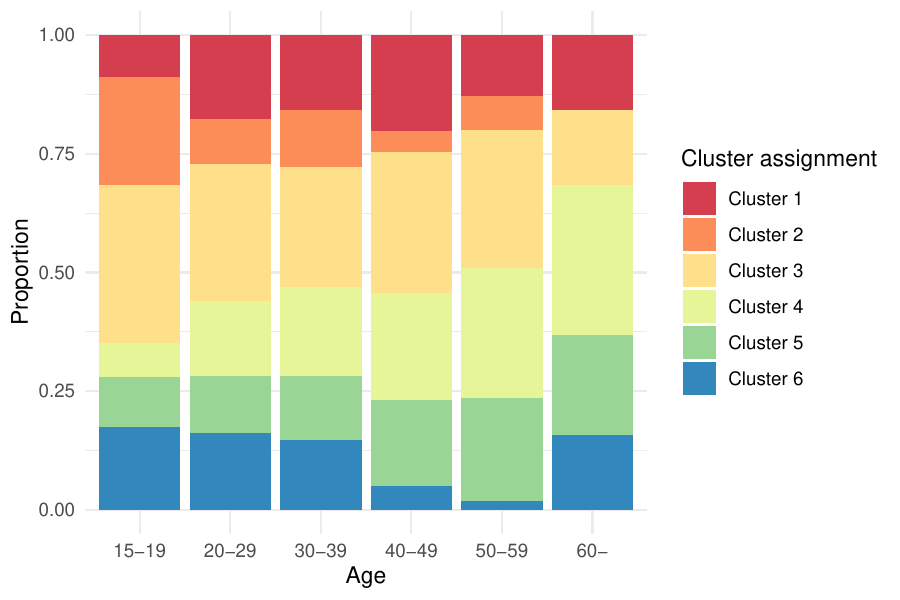}
\endminipage 
\caption{Distribution of covariate V3 according to cluster assignment from run with BMMx with $C=6$ for the sushi dataset.}\label{supp:sushi_cluster_assignment_V3}
\end{figure}

\begin{figure}[H]
\minipage{\textwidth}
\centering
  \includegraphics[width=0.8\linewidth]{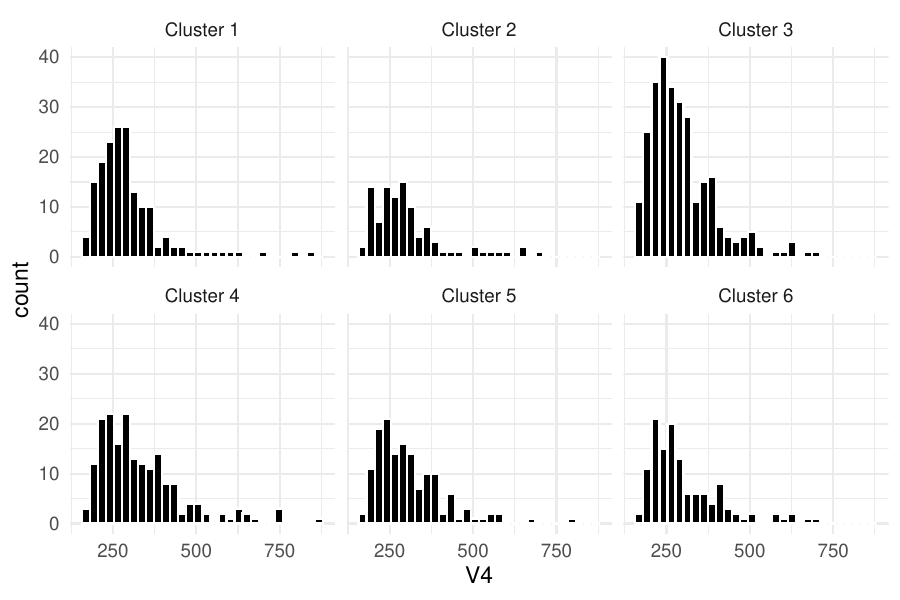}
\endminipage 
\caption{Distribution of covariate V4 according to cluster assignment from run with BMMx with $C=6$ for the sushi dataset.}\label{supp:sushi_cluster_assignment_V4}
\end{figure}

\begin{figure}[H]
\minipage{\textwidth}
\centering
  \includegraphics[width=0.8\linewidth]{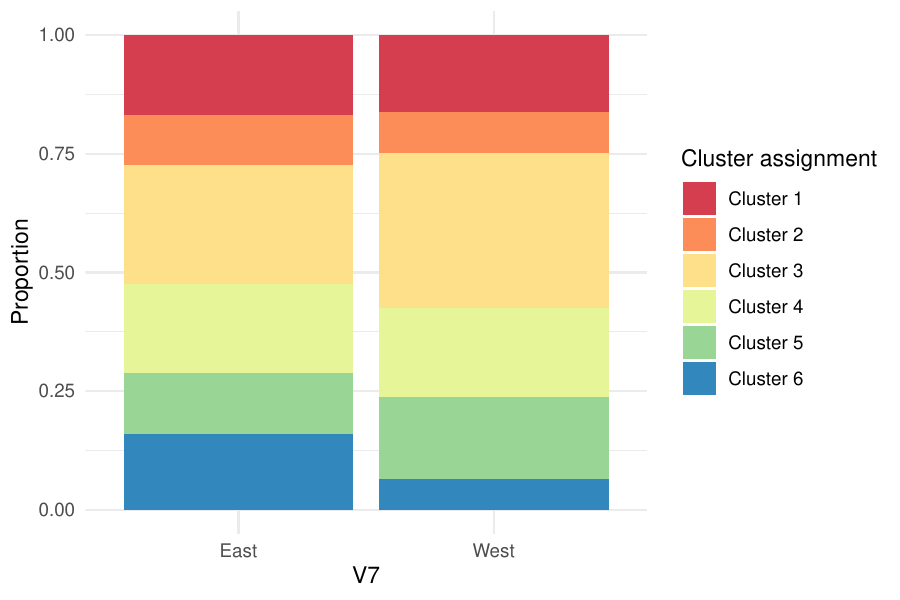}
\endminipage 
\caption{Distribution of covariate V7 according to cluster assignment from run with BMMx with $C=6$ for the sushi dataset.}\label{supp:sushi_cluster_assignment_V7}
\end{figure}

\begin{figure}[H]
\minipage{\textwidth}
\centering
  \includegraphics[width=0.8\linewidth]{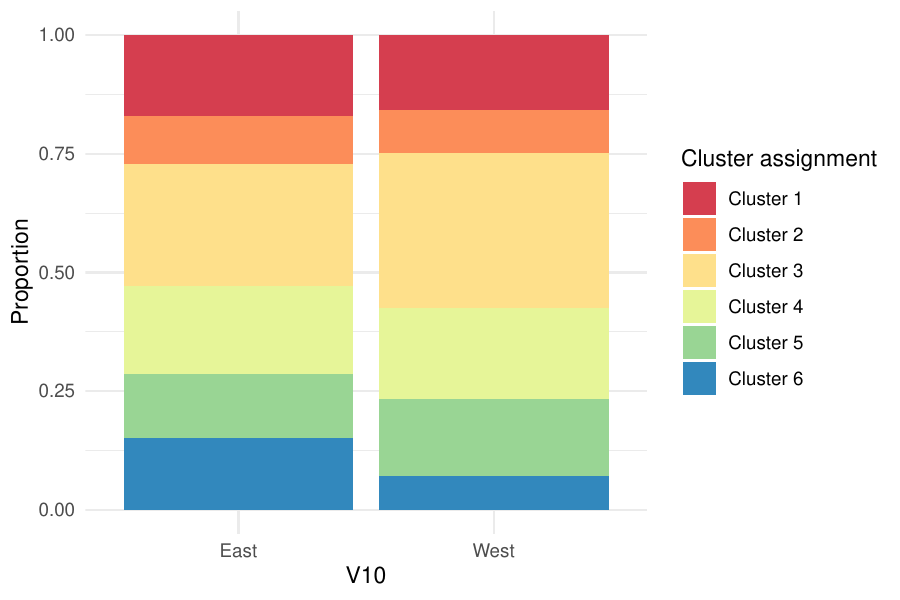}
\endminipage 
\caption{Distribution of covariate V10 according to cluster assignment from run with BMMx with $C=6$ for the sushi dataset.}\label{supp:sushi_cluster_assignment_V10}
\end{figure}

\section{Breast cancer case study}

\begin{figure}[H]
\minipage{\textwidth}
\centering
  \includegraphics[width=0.8\linewidth]{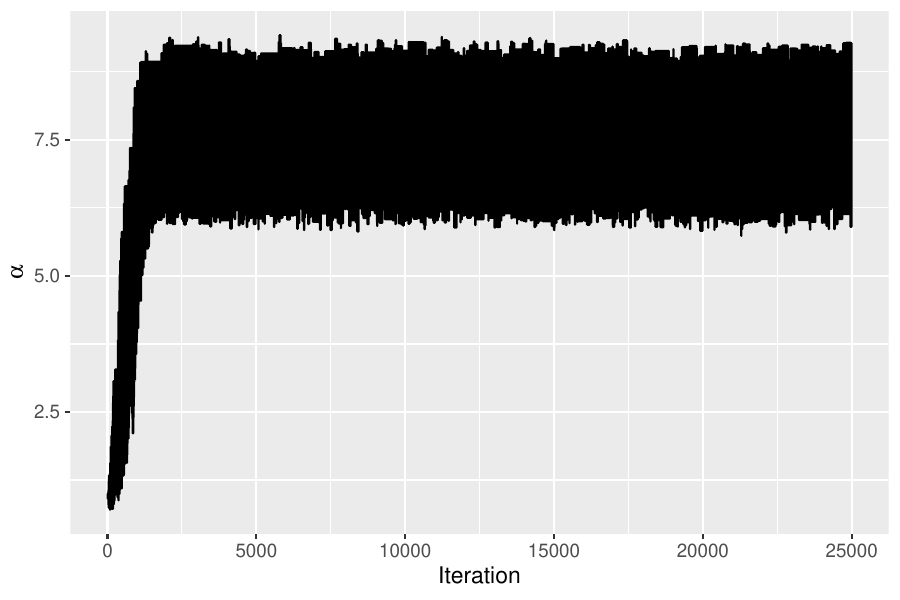}
\endminipage 
\caption{Trace plot of $\alpha$ for the breast cancer dataset, with $C=4$ from run with BMMx.}
\label{supp:BRCA_alpha_convergence}
\end{figure}

\begin{figure}[H]
\minipage{\textwidth}
\centering
  \includegraphics[width=0.8\linewidth]{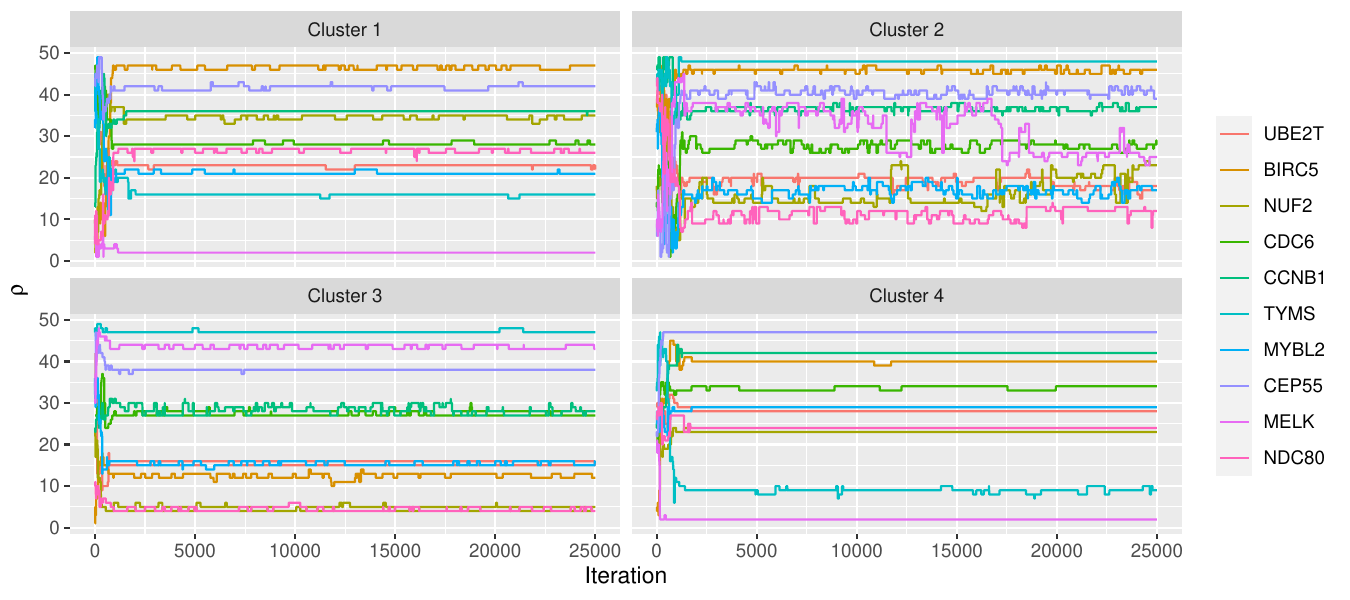}
\endminipage 
\caption{Trace plot of $\mathbf{\rho}_c$ for the breast cancer dataset, with $C=4$ from run with BMMx.}\label{supp:BRCA_rho_convergence}
\end{figure}

\begin{figure}[H]
\minipage{\textwidth}
\centering
  \includegraphics[width=0.8\linewidth]{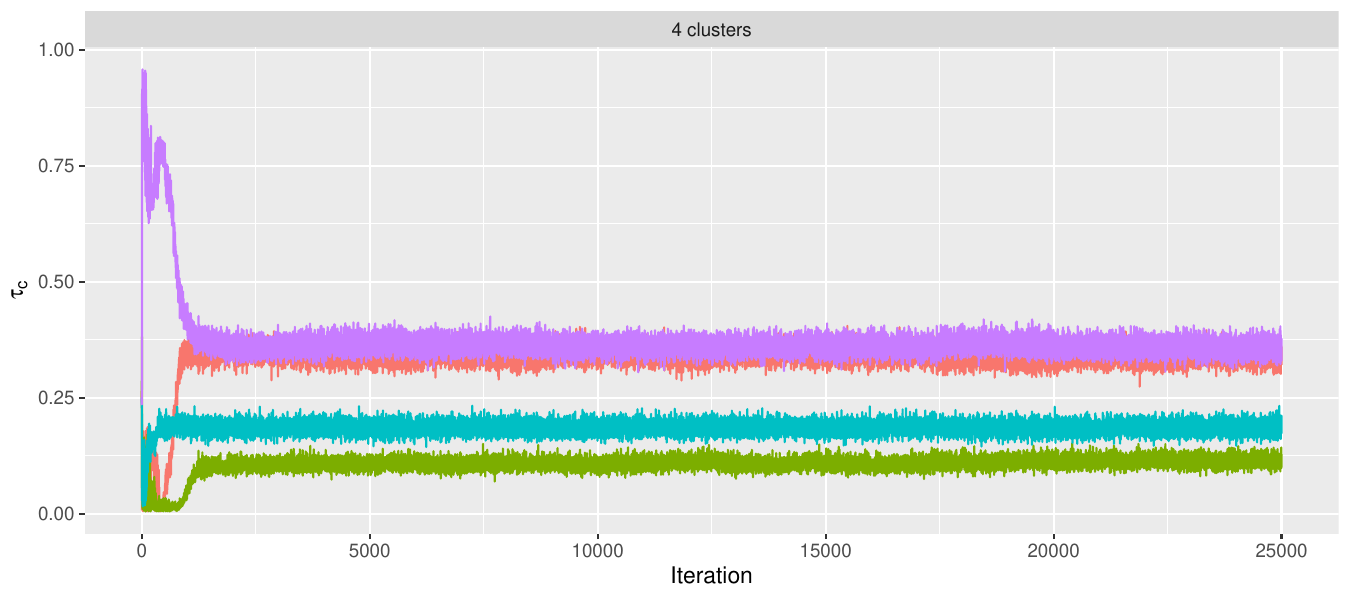}
\endminipage 
\caption{Trace plot of $\mathbf{\tau}_c$ for the breast cancer dataset, with $C=4$ from run with BMMx.}\label{supp:BRCA_cluster_convergence}
\end{figure}

\begin{table}[H]
\scriptsize
\centering
\begin{tabular}{rllll|llll}
  \hline
   & \multicolumn{4}{c|}{BMMx} & \multicolumn{4}{c}{BMM}\\ 
 & Cluster 1 & Cluster 2 & Cluster 3 & Cluster 4 & Cluster 1 & Cluster 2 & Cluster 3 & Cluster 4 \\ 
  \hline
1 & MIA & KRT17 & PTTG1 & MIA & PTTG1 & KRT17 & MIA & MIA \\ 
  2 & MELK & UBE2C & FGFR4 & MELK & FGFR4 & UBE2C & MELK & MELK \\ 
  3 & UBE2C & MIA & KIF2C & CCNE1 & KIF2C & MIA & CCNE1 & UBE2C \\ 
  4 & ACTR3B & ACTR3B & NDC80 & KRT17 & NDC80 & ACTR3B & KRT17 & ACTR3B \\ 
  5 & CCNE1 & CCNE1 & NUF2 & ACTR3B & NUF2 & CCNE1 & ACTR3B & CCNE1 \\ 
  6 & KRT17 & FOXC1 & MKI67 & UBE2C & MKI67 & MELK & UBE2C & KRT17 \\ 
  7 & MLPH & SLC39A6 & KRT17 & MLPH & KRT17 & FOXC1 & MLPH & MLPH \\ 
  8 & TMEM45B & MKI67 & SLC39A6 & TMEM45B & SLC39A6 & MKI67 & TMEM45B & TMEM45B \\ 
  9 & BCL2 & CENPF & PHGDH & TYMS & UBE2C & SLC39A6 & FGFR4 & TYMS \\ 
  10 & MDM2 & BCL2 & MIA & FGFR4 & MIA & BCL2 & TYMS & BCL2 \\ 
  11 & MKI67 & NAT1 & UBE2C & BCL2 & PHGDH & NAT1 & BCL2 & MDM2 \\ 
  12 & SLC39A6 & NDC80 & CENPF & MDM2 & CENPF & CENPF & MDM2 & MKI67 \\ 
  13 & ANLN & ORC6L & BIRC5 & KIF2C & BIRC5 & ORC6L & KIF2C & ANLN \\ 
  14 & KRT5 & NUF2 & CDC20 & ANLN & CDC20 & NDC80 & GRB7 & SLC39A6 \\ 
  15 & ORC6L & FOXA1 & MYBL2 & GRB7 & MYBL2 & FOXA1 & ANLN & KRT5 \\ 
  16 & TYMS & EGFR & UBE2T & PTTG1 & UBE2T & EGFR & PTTG1 & ORC6L \\ 
  17 & PHGDH & MYBL2 & NAT1 & ORC6L & NAT1 & ANLN & ORC6L & PHGDH \\ 
  18 & FOXA1 & ANLN & EGFR & KRT5 & EGFR & KRT5 & KRT5 & FOXA1 \\ 
  19 & EGFR & KRT5 & MMP11 & PHGDH & MMP11 & MYBL2 & PHGDH & EGFR \\ 
  20 & NAT1 & UBE2T & GRB7 & MKI67 & GRB7 & UBE2T & MKI67 & NAT1 \\ 
  21 & MYBL2 & MAPT & ANLN & FOXA1 & ANLN & MMP11 & FOXA1 & MYBL2 \\ 
  22 & CENPF & MMP11 & SFRP1 & SLC39A6 & SFRP1 & CDC20 & SLC39A6 & UBE2T \\ 
  23 & UBE2T & CDC20 & KRT5 & NUF2 & KRT5 & PHGDH & NUF2 & CENPF \\ 
  24 & FOXC1 & PHGDH & CDH3 & NDC80 & CDH3 & MAPT & NDC80 & FOXC1 \\ 
  25 & MMP11 & RRM2 & CXXC5 & EGFR & CXXC5 & NUF2 & EGFR & NDC80 \\ 
  26 & CDC20 & MELK & PGR & FOXC1 & PGR & RRM2 & FOXC1 & MMP11 \\ 
  27 & NDC80 & CDH3 & CDC6 & NAT1 & CDC6 & CDC6 & NAT1 & CDC20 \\ 
  28 & CDC6 & SFRP1 & CCNB1 & UBE2T & ORC6L & SFRP1 & UBE2T & CDC6 \\ 
  29 & RRM2 & CDC6 & KRT14 & MYBL2 & CCNB1 & CDH3 & MYBL2 & RRM2 \\ 
  30 & CDH3 & CXXC5 & ORC6L & CDC20 & KRT14 & PGR & CENPF & CDH3 \\ 
  31 & SFRP1 & PGR & RRM2 & CENPF & RRM2 & CXXC5 & CDC20 & PGR \\ 
  32 & PGR & KIF2C & BCL2 & MMP11 & BCL2 & MDM2 & MMP11 & SFRP1 \\ 
  33 & CXXC5 & MDM2 & BLVRA & BAG1 & BLVRA & TMEM45B & BAG1 & NUF2 \\ 
  34 & KRT14 & FGFR4 & FOXC1 & CDC6 & FOXC1 & BLVRA & CDC6 & CXXC5 \\ 
  35 & NUF2 & BAG1 & MYC & PGR & MYC & KRT14 & PGR & KRT14 \\ 
  36 & CCNB1 & BLVRA & EXO1 & SFRP1 & EXO1 & CCNB1 & SFRP1 & CCNB1 \\ 
  37 & BLVRA & CCNB1 & MDM2 & RRM2 & MDM2 & BAG1 & RRM2 & BLVRA \\ 
  38 & MYC & KRT14 & CEP55 & CDH3 & CEP55 & KIF2C & CDH3 & MYC \\ 
  39 & BAG1 & PTTG1 & CCNE1 & MYC & CCNE1 & MLPH & MYC & PTTG1 \\ 
  40 & PTTG1 & TMEM45B & GPR160 & BIRC5 & GPR160 & FGFR4 & BIRC5 & BAG1 \\ 
  41 & FGFR4 & CEP55 & FOXA1 & CXXC5 & FOXA1 & CEP55 & CXXC5 & FGFR4 \\ 
  42 & CEP55 & MLPH & MAPT & CCNB1 & MAPT & MYC & CCNB1 & KIF2C \\ 
  43 & GPR160 & MYC & TMEM45B & KRT14 & TMEM45B & PTTG1 & KRT14 & GPR160 \\ 
  44 & KIF2C & EXO1 & MELK & BLVRA & MELK & EXO1 & BLVRA & CEP55 \\ 
  45 & MAPT & GRB7 & MLPH & MAPT & MLPH & GPR160 & MAPT & BIRC5 \\ 
  46 & EXO1 & BIRC5 & ACTR3B & GPR160 & ACTR3B & BIRC5 & GPR160 & EXO1 \\ 
  47 & BIRC5 & GPR160 & TYMS & CEP55 & TYMS & GRB7 & CEP55 & MAPT \\ 
  48 & GRB7 & TYMS & BAG1 & EXO1 & BAG1 & TYMS & EXO1 & GRB7 \\ 
  49 & ERBB2 & ERBB2 & ERBB2 & ERBB2 & ERBB2 & ERBB2 & ERBB2 & ERBB2 \\ 
   \hline
\end{tabular}
\caption{Results of the breast cancer data experiment with $C=4$. CP consensus for each of the clusters with BMMx and BMM.}\label{supp:BRCA_CP_consensus}
\end{table}

\begin{figure}[H]
\minipage{\textwidth}
\centering
\includegraphics[width=0.8\linewidth]{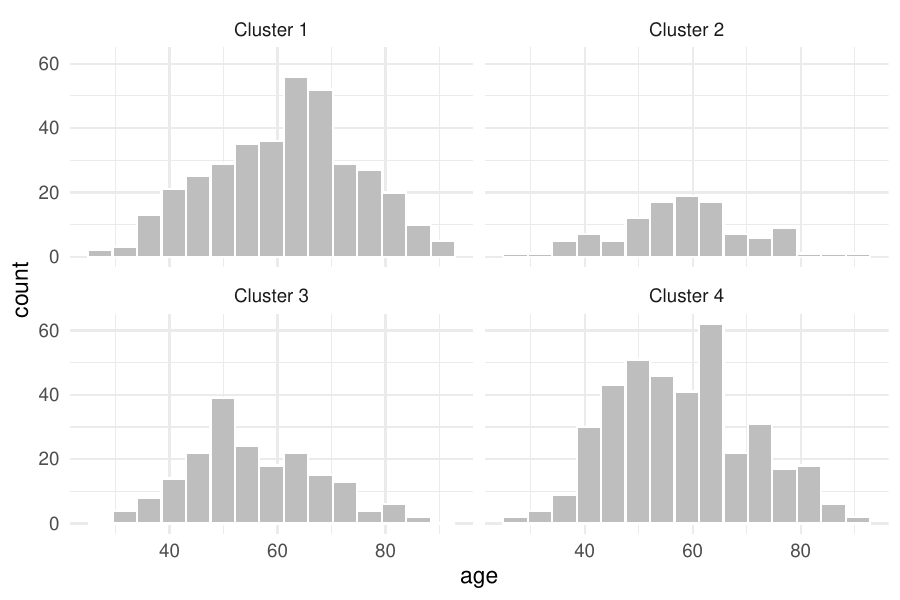}
\endminipage 
\caption{Distribution of covariate "age" according to cluster assignment from run with BMMx with $C=4$ for the breast cancer dataset.}
\label{supp:BRCA_cluster_assignment_age}
\end{figure}

\begin{figure}[H]
\minipage{\textwidth}
\centering
  \includegraphics[width=0.8\linewidth]{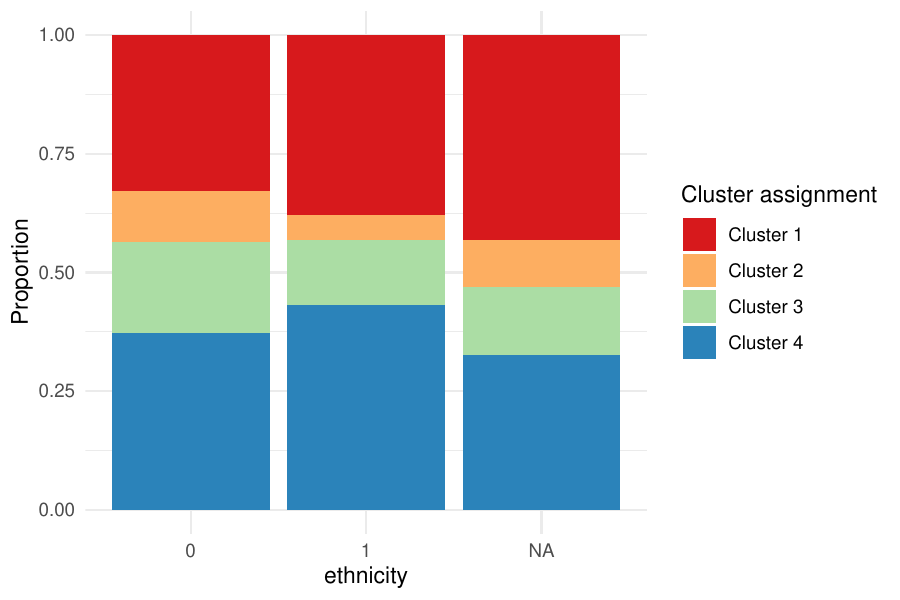}
\endminipage 
\caption{Distribution of covariate "ethnicity" according to cluster assignment from run with BMMx with $C=4$ for the breast cancer dataset.}
\label{supp:BRCA_cluster_assignment_ethnicity}
\end{figure}

\begin{figure}[H]
\minipage{\textwidth}
\centering
  \includegraphics[width=0.8\linewidth]{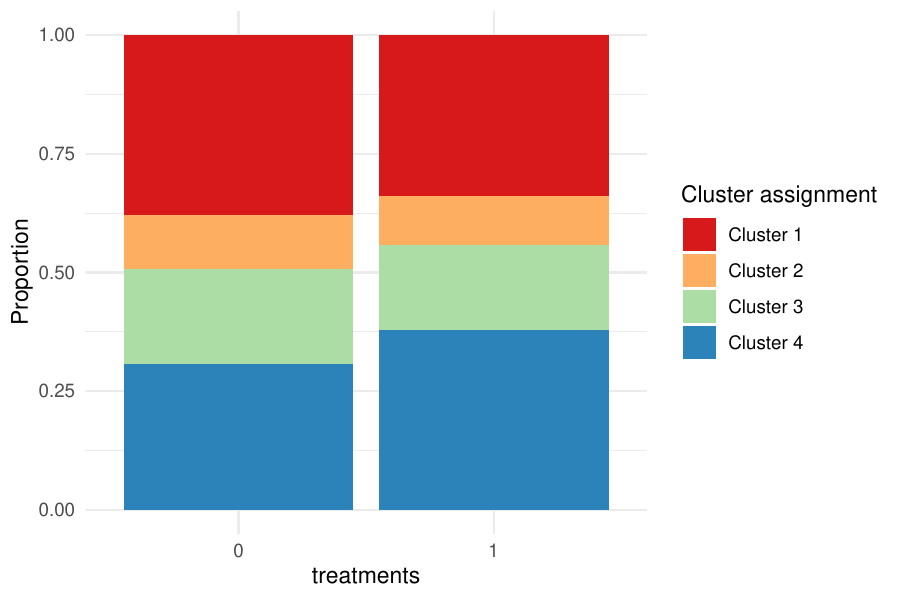}
\endminipage 
\caption{Distribution of covariate "treatments" according to cluster assignment from run with BMMx with $C=4$ for the breast cancer dataset.}
\label{supp:BRCA_cluster_assignment_treatments}
\end{figure}

\begin{figure}[H]
\minipage{\textwidth}
\centering
  \includegraphics[width=0.8\linewidth]{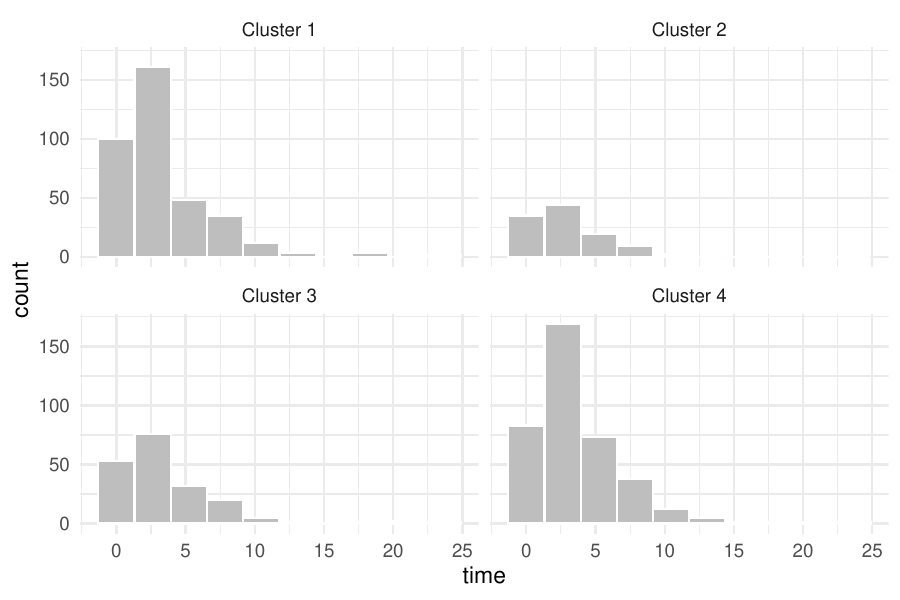}
\endminipage 
\caption{Distribution of covariate "time" according to cluster assignment from run with BMMx with $C=4$ for the breast cancer dataset.}\label{supp:BRCA_cluster_assignment_time}
\end{figure}

\begin{figure}[H]
\minipage{\textwidth}
\centering
  \includegraphics[width=0.8\linewidth]{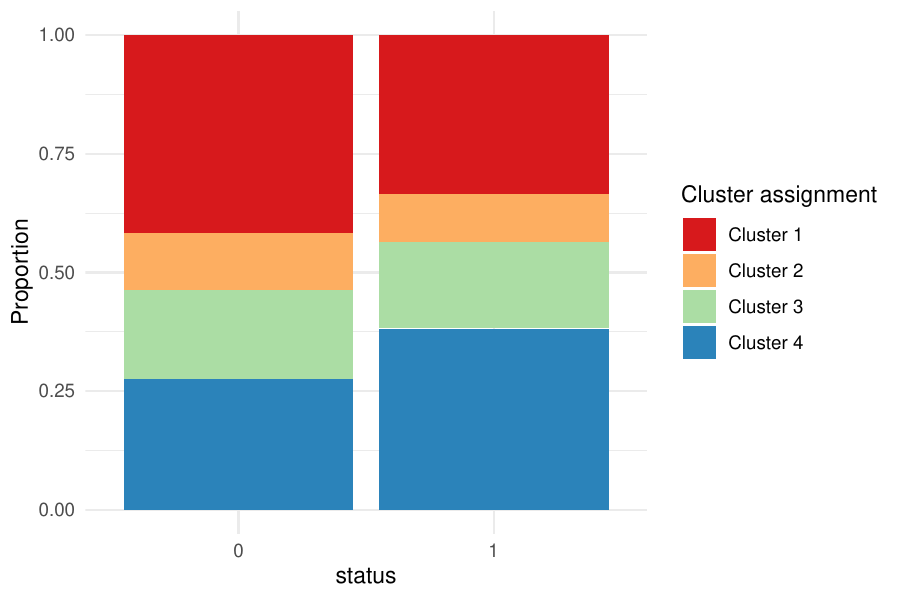}
\endminipage 
\caption{Distribution of covariate "status" according to cluster assignment from run with BMMx with $C=4$ for the breast cancer dataset.}\label{supp:BRCA_cluster_assignment_status}
\end{figure}

\end{document}